\documentclass[twoside,12pt,reqno]{amsart}

  \usepackage{amsaddr}

\makeatletter

\usepackage{caption}
\usepackage{subcaption}

\usepackage{setspace}
\usepackage{amscd,amssymb,amsmath,latexsym}
\usepackage[mathscr]{euscript}
\usepackage{epsfig}
\usepackage{xcolor} 

\usepackage{bm} 
\allowdisplaybreaks 

\usepackage[left=1in,right=1in,top=1in,bottom=1in]{geometry} 
\usepackage{physics} 

\usepackage{qcircuit}
\usepackage[ruled,vlined]{algorithm2e}
\usepackage{tcolorbox}
\usepackage[linguistics]{forest} 
\usepackage{hyperref}

\usepackage{todonotes}

\usepackage{algpseudocode}

\newcommand{\N}{\mathbb N}

\newcommand{\C}{\mathbb C}

\newcommand{\U}{\mathbb U}

\newcommand{\tTr}{\text{Tr}\,}        
\newcommand{\tDim}{\text{dim}\,}      
\newcommand{\tL}{\text{L}\,}          
\newcommand{\tD}{\text{D}\,}          

\newcommand{\cH}{\mathcal H}
\newcommand{\cC}{\mathcal C}
\newcommand{\cD}{\mathcal D}
\newcommand{\cB}{\mathcal B}

\newcommand{\cA}{\mathcal A}

\newcommand{\cW}{\mathcal W}

\newcommand{\cU}{\mathcal{U}}

\newcommand{\sumi}{\displaystyle \sum_{i=0}^{\infty}} 

\newcommand{\variance}{\text{Var}}
\newcommand{\dt}{\frac{\partial}{\partial\theta}}
\newcommand{\wg}{\text{Wg}}
\newcommand{\avg}{\mathbb{E}}

\newtheorem{Thm}{Theorem}[section]

\newtheorem{Lem}[Thm]{Lemma}
\newtheorem{Prop}[Thm]{Proposition}
\newtheorem{Rmk}[Thm]{Remark}
\newtheorem{Ex}[Thm]{Example}
\newtheorem{Def}[Thm]{Definition}
\newtheorem{Term}[Thm]{Terminology}
\newtheorem{Con}[Thm]{Convention}

\title[Variational Entanglement Detection]{A Variational Quantum Algorithm For Approximating Convex Roofs}

\author{George Androulakis and Ryan McGaha}
\address{Department of Mathematics, University of South Carolina, 1523 Greene St, Columbia, SC 29208, USA}

\email{giorgis@math.sc.edu}
\email{rmcgaha@email.sc.edu}

\thanks{The work presented here is part of the Ph.D. thesis of the second author which is conducted under 
the supervision of the first author for the applied mathematics doctoral program at the University of 
South Carolina. As such, results herein may reappear in the PhD dissertation of the second author.}

\begin{document}
\maketitle

\begin{abstract}
    Many entanglement measures are first defined for pure states of a bipartite Hilbert space, and then extended to mixed states via the convex roof extension. 
    In this article we alter the convex roof extension of an entanglement measure, to produce a sequence of extensions 
    that we call $f$-$d$ extensions, for $d \in \mathbb{N}$, where  $f:[0,1]\to [0, \infty)$ is a fixed continuous function which vanishes only at zero.  
    We prove that for any such function $f$,
    and any continuous, faithful, non-negative function, (such as an entanglement measure), $\mu$  on the set of pure states of a finite dimensional bipartite Hilbert space, the collection of 
    $f$-$d$ extensions of  $\mu$ detects entanglement, i.e. a mixed state $\rho$ on a finite dimensional bipartite Hilbert space is separable, 
    if and only if 
    there exists $d \in \mathbb{N}$ such that the $f$-$d$ extension of $\mu$ applied to $\rho$ is equal to zero.
    We introduce a quantum variational algorithm which aims to approximate the $f$-$d$ extensions of entanglement measures
    defined on pure states. However, the algorithm does have its drawbacks. We show that this algorithm  exhibits barren plateaus when 
    used to approximate the family of $f$-$d$ extensions of the Tsallis entanglement entropy for a certain function $f$ and unitary ansatz $U(\theta)$ of sufficient depth. In practice, if additional information about the state is known, then one needs to avoid using the suggested
ansatz for long depth of circuits.
\end{abstract}

\section{Introduction}
The detection and quantification of quantum entanglement is a fundamental problem in quantum information theory. A common method for such a quantification is via the use of entanglement measures and entanglement monotones \cite{HHHH, PV05, VPRK, V98}. However, entanglement measures are often difficult to compute for arbitrary density operators, as many of these measures are defined via the convex roof construction. There are some classical numerical algorithms \cite{roth} for computing the convex roof of entanglement measures, and even some quantum variational algorithms \cite{wcubed} for computing the logarithmic negativity of arbitrary densities. In this article 
we define a family of extensions of entanglement measures defined on pure states, and we introduce a quantum variational algorithm 
which aims to approximate these extensions. We prove that the family of certain extensions of the Tsallis entanglement entropy 
exhibits barrens plateaus. 

First, let us review the definitions of entanglement measures and the convex roof construction in general. 
All Hilbert spaces that we consider throughout this article are finite dimensional. Also, if $\cH$ is a Hilbert space
then $\tD(\cH)$ will denote the set of all density operators on $\cH$, (i.e.\ positive semidefinite operators, of trace equal to $1$).

\begin{Def} \label{Def:entanglement_measure}
Let $\cH$ be a finite dimensional bipartite Hilbert space, i.e.\ $\cH=\cA\otimes\cB$ for some finite dimensional Hilbert spaces $\cA$ and $\cB$. Then a function $\mu:\tD(\cH)\rightarrow [0,\infty)$ is called an entanglement measure if the following three properties hold for all $\rho\in\tD(\cH)$:

\begin{enumerate}
    \item $\mu(\rho)=0$ if and only if $\rho$ is separable, (i.e. $\mu$ is \emph{faithful}).\\
    \item If $\Lambda$ is a channel of Local Operations and Classical Communications, 
    (LOCC channel in short), $\rho \in \tD(\cH)$, and $\Lambda (\rho) = \sum_i p_i \rho_i$ is a channel  for some convex coefficients $(p_i)_i$ and some states $(\rho_i)_i)$, then
    $\sum_i p_i \mu(\rho_i) \leq \mu(\rho)$. While the structure of LOCC channels is complicated in general, a thorough overview 
    of LOCC channels can be found in \cite{LOCC}. This property is referred as \lq\lq LOCC monotonicity\rq\rq\ or 
    \lq\lq decrease on average under an LOCC channel\rq\rq).\\ 
    \item $\mu((U_1 \otimes U_2) \rho (U_1^*\otimes U_2^*))=\mu(\rho)$ where $U_1$ and $U_2$ are unitaries on 
    $\cA$ and $\cB$ respectively, (i.e.\ $\mu$ is invariant under local unitaries).
\end{enumerate}
If $\mu$ is defined only on the set of pure states of $\cH$ and satisfies the above properties $(1)$-$(3)$ for all pure states 
$\rho$ of $\cH$, then we say that $\mu$ is an entanglement measure on the set of pure states of $\cH$. If $\mu$ is defined only 
on the set of pure states of $\cH$ and satisfies the above property~$(1)$ for all pure states $\rho$ of $\cH$, then we say that 
$\mu$ is a faithful non-negative function on the set of pure states.
\end{Def}

If $\cH$ is a bipartite Hilbert space and $\mu$ is an entanglement measure on the pure states of $\cH$, then 
one extends the domain of $\mu$ to the set $\tD(\cH)$ of density operators on $\cH$, 
using the \emph{convex roof extension} which is defined by
\begin{equation} \label{eqn : convexroof}
     \mu(\rho) = \inf 
    \left\{ \sumi p_i \mu(\psi_i) : \{ \{ p_i \}_{i=1}^\infty,  \{\psi_i\}_{i=0}^\infty \} \in \cC\cD (\rho)
    \right\}
\end{equation}
where $\cC\cD (\rho)$ is the set of all convex decompositions of $\rho$ with respect to pure states, 
often called \emph{pure state ensembles}. More precisely, $\cC\cD (\rho)$ is defined by 
\begin{equation}
    \cC\cD (\rho) = \left\{ \{ \{ p_i \}_{i=1}^\infty,  \{\ket{\psi_i}\}_{i=0}^\infty \} : p_i \geq 0, 
    \sum_{i=1}^\infty p_i =1, \ket{\psi_i} \in \cH, 
    \| \ket{\psi_i}\|=1, \sumi p_i \ketbra{\psi_i} = \rho \right\}
\end{equation}

It can be shown that the function $\mu$ given by Equation~\eqref{eqn : convexroof} is the largest convex function defined on the 
set of all states of $\cH$, which is less than 
or equal to the function $\mu$ on the set of pure states of $\cH$.
While at first glance it's not clear that the infimum in Equation~\eqref{eqn : convexroof} is attained, it has indeed been proven 
\cite{shiro, us} that the infimum is attained as long as the measure $\mu$ is norm continuous on the set of unit vectors of $\cH$,
(i.e.\ on the set of pure states of $\cH$). 

\begin{Term}
The elements of $\cC\cD (\rho)$ 
at which the infimum is achieved are called optimal pure state ensembles or simply OPSEs. 
\end{Term}

Probably the most prevalent example of an entanglement measure formed in this way is the \emph{entanglement of formation}, which is first defined on pure states by 
\begin{equation}
\text{S}(\ketbra{\psi})=-\tTr\big( \tTr_\cA(\ketbra\psi) \log (\tTr_\cA \ketbra\psi) \big) =
-\tTr\big( \tTr_\cB(\ketbra\psi) \log (\tTr_\cB \ketbra\psi) \big)
\end{equation}
This definition is then extended to the set of all density operators by the convex roof extension method described in 
Equation~\eqref{eqn : convexroof}. It has been proven that the entanglement entropy is norm continuous on $\cH$ \cite{FAN}, and so the infimum, (OPSE), is always achieved in this case. 

Another example of entanglement measurement formed in the same way is the Tsallis entanglement entropy $T_2$, which is defined by 
\begin{equation} \label{Tsallis}
    T_2 (\ketbra{\psi})= 1 - \tTr ((\tTr_{\cA} (\ketbra{\psi}))^2)
\end{equation}
for all pure states $\ket\psi \in \cH = \cA \otimes \cB$. It is then extended to general mixed states $\rho$ of $\cH$ via the convex roof construction given in Equation~\eqref{eqn : convexroof}, namely
\begin{equation}
    T_2(\rho) = \inf 
    \left\{ \sumi p_i T_2(\ketbra{\psi_i}) : \{ \{ p_i \}_{i=0}^\infty,  \{\psi_i\}_{i=0}^\infty \} \in \cC\cD (\rho)
    \right\}
\end{equation}

The natural question that arises is how to find the optimal pure state ensembles that yield the infimum of the 
extension of an entanglement measure which is defined on the pure state of a bipartite Hilbert space and is extended to
the mixed states using the convex roof extension. This is an important question because 
once the OPSEs are known, then one can compute precisely the entanglement measure.

In this article we give a variational algorithm which aims at approximating the OPSE. It is well known that a main difficulty
in the convergence of variational optimization methods is the existence of \emph{Barren Plateaus} \cite{bp1,bp2,bp3,bp4}. 
Unfortunately, even for 
\lq\lq simple\rq\rq\ entanglement measures, such as the Tsallis entanglement measure, to determine whether or not a variational algorithm
has barren plateaus, is not an easy task. 
In this article we study barren plateaus of a
variant of the convex roof extension method of the Tsallis entanglement 
entropy, 
instead of 
studying barren plateaus of the Tsallis entanglement entropy itself, since the latter is significantly more technical.
Nevertheless in Section~\ref{sec:Swap_test} we indicate how one can use the ideas presented here in order to examine barren plateaus of the Tsallis entanglement entropy. Additionally, we believe that our method adapted for the Tsallis entanglement entropy will only predict barren plateaus under
 stronger assumptions on the depth of the ansatz than the depth assumptions used in this article. In practice, one will have to avoid using the ansatz
that we use this article for circuits of large depth. This can be possible
if additional information is known about the state whose entanglement is studied.
The variant of the convex roof extension method that we use here 
does not yield an entanglement measure, but it is still  allows us to detect 
entanglement. It depends on a fixed function $f:[0,1]\to [0, \infty)$ which vanishes only at zero,
and it yields a family of extensions indexed by the integers $d\in \N$. The domains of these extensions increase with $d$,
and they become equal to the set of density operators on the Hilbert space when $d$ is large enough, (more precisely, when $d$ 
strictly exceeds the dimension of the manifold of the density operators on the given bipartite Hilbert space). 
These extensions decrease with respect to $d\in \N$ at any fixed density operator, and their infimum for all $d$ is equal to zero. 
More precisely, we have the following definition.

\begin{Def} \label{Def:f-d-extension}
Let $\cH$ be a bipartite Hilbert space, $\mu$ be an entanglement measure on the set of pure states of $\cH$, and $f:[0,1]\to [0 ,\infty)$
be a function that vanishes only at $0$.
For every $d \in \N$ define the set $\tD(\cH)_d$ of density operators of $\cH$ that can be written as a convex combination of at most $d$ 
many pure states, i.e.\ 
\begin{align} \label{DHt}
\begin{split}
\tD(\cH)_d=  \bigg\{ \rho \in \tD(\cH) &: \text{ there exists a family of pure states }(\ketbra{\psi_i}{\psi_i})_{i=1}^d\\
& \text{ and }(p_i)_{i=1}^d \subseteq [0,1] 
 \text{ with } \sum_{i=1}^d p_i =1 \text{ and } \rho =\sum_{i=1}^d p_i \ketbra{\psi_i} \bigg\}.
\end{split}
\end{align}
Note that 
\begin{enumerate}
\item[(i)] $\tD(\cH)_1$ is equal to the set of pure states of $\cH$, (in this case, $p_1=1$).
\item[(ii)] $\tD(\cH)_d \subseteq \tD(\cH)_{d+1}$ for every $d \in \N$, (since the $p_i$'s are allowed to be equal to zero).
\item[(iii)] There exists $d \in \N$ such that $\tD(\cH)_d= \tD(\cH)$, (indeed, by Caratheodory's theorem in convex analysis, this happens when $d \geq \text{dim}\, (\tD(\cH))+1$).
\end{enumerate}
Define a function $\mu_{f,d} : \tD(\cH)_d \to [0,\infty)$ by 
$$
\mu_{f,d} (\rho)= \inf \bigg\{ \sum_{i=1}^d f(p_i) \mu (\ketbra{\psi_i}): \rho = \sum_{i=1}^d p_i \ketbra{\psi_i} \in \tD(\cH)_d \text{ as in Equation~\eqref{DHt}} \bigg\}.
$$
We call the function $\mu_{f,d}$  the \lq\lq $f$-$d$ extension of $\mu$\rq\rq.\ We call the sequence $(\mu_{f,d})_{d \in \N}$
the \lq\lq sequence of $f$-extensions of $\mu$\rq\rq.\ 
\end{Def}

\begin{Rmk} \label{Rmk:infimum_attained}
Let $\cH$ be a finite dimensional bipartite Hilbert space, $\mu$ be a continuous, non-negative function on the set of pure states of $\cH$,
and 
$f:[0,1]\to [0,\infty)$ be a continuous function which vanishes only at $0$. Then, for every $d \in \N$ the infimum in the 
definition of $\mu_{f,d}$ is attained, (i.e.\ it is a minimum).
\end{Rmk}

Indeed, since $\cH$ is finite dimensional, the set of pure states on $\cH$ is a compact set. Let $\mathcal{P}$ denote the set
of pure states on $\cH$. Also, the set 
$$
\mathcal{C}= \bigg\{ (p_i)_{i=1}^d : p_i \geq 0 \text{ for all }i=1,\ldots , d \text{ and } \sum_{i=1}^d p_i=1 \bigg\}
$$
is compact as well. Now the function defined on $\mathcal{P}^d \times \mathcal{C}$ by
$$
\mathcal{P}^d \times \mathcal{C} \ni (\ketbra{\psi_i}{\psi_i})_{i=1}^d \times (p_i)_{i=1}^d \mapsto 
\sum_{i=1}^d f(p_i)\mu (\ketbra{\psi_i}{\psi_i}),
$$
is a continuous function, since both $f$ and $\mu$ are assumed to be continuous. Hence it achieves its minimum on its compact domain,
which finishes the proof of Remark~\ref{Rmk:infimum_attained}.

Remark~\ref{Rmk:infimum_attained} justifies the next terminology.

\begin{Term} \label{Term:OPSE}
Let $\cH$ be a finite dimensional bipartite Hilbert space, $\mu$ be a continuous non-negative function on the set of pure states of $\cH$,
and $f:[0,1]\to [0, \infty)$ be a continuous function which vanishes only at $0$. Let  $d \in \N$, and $\rho \in \tD(\cH)_d$.
The tuple 
$((p_i)_{i=1}^d, (\ketbra{\psi_i}{\psi_i})_{i=1}^d)$ with $p_i \geq 0$ for all $i=1, \ldots , d$, $\sum_{i=1}^d p_i =1$,
$\ketbra{\psi_i}{\psi_i}$ being pure states on $\cH$, and $\mu_{f,d}(\rho) = \sum_{i=1}^d f(p_i)\mu (\ketbra{\psi_i}{\psi_i})$,
is called \lq\lq Optimal Pure State Ensemble (OPSE) for $\mu$, $f$ and $d$\rq\rq.\ 
\end{Term}

In this article we study the sequence of $f$-extensions of Tsallis 
entanglement entropy $T_2$, where $f:[0,1]\to [0,\infty)$ is defined by $f(x)=x^2$.
We will denote by $(T_{f,d})_{d \in \N}$ the sequence of $f$-extensions of Tsallis entanglement entropy $T_2$ defined on the set of 
pure states of a bipartite Hilbert space $\cH= \cA \otimes \cB$. More precisely, by combining Equation~\eqref{Tsallis} and 
Definition~\ref{Def:f-d-extension} we obtain

\begin{equation} \label{E:objective_function}
T_{f,d}(\rho) = \inf \left\{ \sum_{i=1}^d  q_i^2 T_2(\phi_i) : 
\rho= \sum_{i=1}^d p_i \ketbra{\psi_i} \in \tD(\cH)_d \text{ as in \eqref{DHt}}
\right\}.
\end{equation}

\begin{Rmk} \label{infimum_is_zero}
The reason that we bound the length of the convex combinations in Definition~\ref{Def:f-d-extension} by a finite number $d$,
is because otherwise the infimum in Equation~\eqref{E:objective_function} would be equal to zero for every $\rho$.
\end{Rmk}

Indeed, in order to verify the statement of Remark~\ref{infimum_is_zero}, notice that if 
$\rho = \sum_{i=1}^d p_i \ketbra{\psi_i}{\psi_i} \in \tD(\cH)_d$ for some 
$(p_i)_{i=1}^d \subset [0,1]$ with $\sum_{i=1}^d p_i =1$ and a sequence of pure states $(\ketbra{\psi_i}{\psi_i})_{i=1}^d$, 
then for every $n \in \N$
we can write $\rho = \sum_{i=1}^d \sum_{j=1}^n \frac{p_i}{n} \ketbra{\psi_i}$ (i.e.\ we divide each $p_i\ketbra{\psi_i}{\psi_i}$ by $n$ and we repeat $n$ times), and 
$$
\sum_{i=1}^d \sum_{j=1}^n \left(\frac{p_i}{n}\right)^2 \mu(\ketbra{\psi_i})=n \sum_{i=1}^d \frac{p_i^2}{n^2} \mu(\ketbra{\psi_i}) 
= \frac{1}{n} \sum_{i=1}^d p_i^2 \mu(\ketbra{\psi_i}),
$$
which tends to zero as $n$ tends to infinity.

Given an entanglement measure $\mu$ on the set of pure states of a finite dimensional bipartite Hilbert space,
and a continuous function $f:[0,1] \to [0, \infty)$ which vanishes only at $0$, the 
sequence of the $f$-extensions of $\mu$ can be used in order to detect entanglement in the following sense:

\begin{Def}
Let $\cH$ be a bipartite Hilbert space, and for every $d \in \N$ let $\tD(\cH)_d$ be a subset of the set $\tD(\cH)$ of density 
matrices of $\cH$, satisfying properties $(i)$, $(ii)$ and $(iii)$ of Definition~\ref{Def:f-d-extension}.
For every $d \in \N$ let a function $\mu_d: \tD(\cH)_d \to [0,\infty)$. We say that the family 
$(\mu_d)_{d \in \N}$ detects entanglement, if for every density operator $\rho$ on $\cH$,
$\rho$ is separable if and only if there exists $d \in \N$ such that $\rho \in \tD(\cH)_d$ and $\mu_d (\rho)=0$. 
\end{Def}

\begin{Prop}\label{prop:faithful}
Let $\mu$ be any continuous, faithful, non-negative function on the set of pure states of a finite dimensional bipartite Hilbert space
$\cH= \cA \otimes \cB$. Let $f:[0,1]\to [0,\infty)$ be  a continuous function that vanishes only at zero.
Then, the $f$-$d$ extensions $(\mu_{f,d})_{d \in \N}$ of $\mu$ to the density operators of $\cH$, detect entanglement, as in the previous definition.
\end{Prop}

\begin{proof}
 $(\Leftarrow )$ Let $\rho$ be a mixed state on $\cH$ such that $\mu_{f,d}(\rho)=0$ for some $d \in \N$.  Then  by Remark~\ref{Rmk:infimum_attained}, there exists some pure state ensemble $\big\{(p_i)_{i=1}^d,(\ketbra{\psi_i})_{i=1}^d\big\}$ with $\sum_{i=1}^dp_i\ketbra{\psi_i}=\rho$ and 
\begin{equation}
    \mu_{f,d}(\rho) = \sum_{i=1}^df(p_i)\mu(\ketbra{\psi_i}{\psi_i}) =0.
\end{equation}
If we assume without loss of generality that each $p_i>0$, then it must follow that 
$\mu(\ketbra{\psi_i}{\psi_i})=0$ for each $i\in\{1,\dots,d\}$. Since $\mu$ is a faithful, non-negative function on the set of pure states,
it must be true that $\ketbra{\psi_i}{\psi_i}$ is factorable for each $i\in\{1,\dots,d\}$. Thus $\rho$ can be written as a convex combination of factorable pure states, and is therefore separable.

$(\Rightarrow )$ Conversely, suppose that $\mu$ is separable. Then there exists some ensemble $\big\{(p_i)_{i=1}^d$, $(\ketbra{\psi_i}{\psi_i})_{i=1}^d\big\}$ of $\rho$ such that each $\ketbra{\psi_i}{\psi_i}$ is factorable. 
Therefore $\mu(\ketbra{\psi_i}{\psi_i})=0$ for each $i\in\{1,\dots,d\}$ and so 
\begin{equation}
    \mu_{f,d}(\rho) = \sum_{i=1}^df(p_i)\mu(\ketbra{\psi_i}{\psi_i}) =0.
\end{equation}
\end{proof}

\section{The Algorithm} \label{sec: algorithm}

It is difficult to compute optimal pure state ensebles in Equation~\eqref{eqn : convexroof} or in Terminology~\ref{Term:OPSE}, because the proof
for their existence, (given in Remark~\ref{Rmk:infimum_attained}), is not constructive. However, methods for approximating them on a quantum computer are heavily implied by the following theorem \cite{Schrod,HJW,NH,Merm,NG}. 

\begin{Thm}[Purification Theorem]
Let $\cH$ be a finite dimensional Hilbert space and suppose that $\rho\in\tD(\cH)$ can be expressed as both of the following pure state ensembles:
$$
\sum_{i=1}^d p_i \ketbra{\psi_i}{\psi_i}= \rho = \sum_{i=1}^d q_i\ketbra{\varphi_i}{\varphi_i}
$$
and define
\begin{equation}\label{eqn : psi}
\ket\psi = \sum_{i=1}^d \sqrt{p_i}\ket{\psi_i}\otimes \ket{i}
\end{equation}
where $\{\ket{i}\}_{i=1}^d$ is an orthonormal basis for $\cW:=\C^d$. Then there exists a unitary $U\in\U(d)$ such that $(I_\cH\otimes U)\ket\psi = \sum_{i=1}^d \sqrt{q_i}\ket{\varphi_i}\otimes\ket{i}$.
\end{Thm}

This theorem tells us that any two pure state ensembles of the same length are related by a unitary map. 
So if we already know one pure state ensemble of a given density, we can find all others simply by applying unitaries 
on the ancilla space of a purification $\ket\psi$ of $\rho$. Hence we can minimize the quantity 
$\sum_{i=1}^d q_i \mu(\varphi_i)$ over the unitary group, where the $q_i$ and $\varphi_i$ are as in the previous theorem, and $\mu$ is an arbitrary entanglement measure. Moreover, the theorem hints at using the following quantum circuit to achieve our goal. 

\begin{figure}[h!]
\captionsetup{width=3.8 in}
\centering
\includegraphics[width =3.8 in]{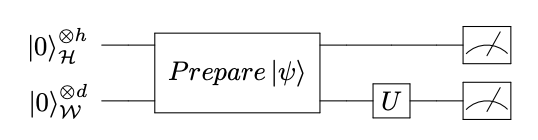} 
\caption{The quantum circuit implied by the purification theorem.}
\label{circ:1}
\end{figure}

Here $h=\tDim(\cH)$ and $\rho\in D(\cH)$, while $d = \tDim(\cW)$ where $\cW$ is the ancilla space. In particular, $d$ is also the length of the pure state decompositions output from the circuit. Usually $d$ is taken to be greater than $h$. 

Examining the circuit, notice that after preparing the system in the state $\ket\psi$ and applying the unitary $U$, the system 
is put into the superposition $\sum_{i=1}^d \sqrt{q_i}\ket{\varphi_i}\otimes\ket{i}$. When measuring the system, if the state 
$\ket{j}$ is observed on the ancilla space, then it must be the case that the top qubits are in the state $\ket{\varphi_j}$ 
with certainty. Running this circuit a number of times allows the observer to experimentally reconstruct both $\ket{\varphi_j}$ and $q_j$, and therefore compute the quantity $\sum_{i=1}^d q_i \mu(\varphi_i)$. In order to find the optimal pure state ensemble we need 
expressions for $\ket{\varphi_i}$ and $q_i$ in terms of $U$. 
By examining the circuit step by step, we see that the circuit first prepares the purification $\ket\psi$. The unitary $U$ is then 
applied to the ancilla space resulting in the state $(I_\cH\otimes U) \ketbra{\psi}{\psi} (I_\cH\otimes U^*)$. Lastly, measuring 
the ancilla yields the states
\begin{equation}
    \frac{1}{q_i} \Big( I_\cH\otimes\ketbra{i}{i}U\Big) \ketbra{\psi}{\psi} \Big( I_\cH \otimes U^*\ketbra{i}{i}\Big)
\end{equation}

where
\begin{equation}\label{eqn : qi}
\begin{split}
    q_i &= \tTr\Bigg( \Big(I_\cH\otimes\ketbra{i}{i}U\Big)\ketbra{\psi}{\psi} \Big(I_\cH\otimes U^*\ketbra{i}{i}\Big)\Bigg)\\
    &= \bra\psi \Big( I_{\cH}\otimes U^*\ketbra{i}{i}U\Big)\ket\psi.\\
\end{split}
\end{equation}
And to access the state on the space $\cH$, we partial trace the ancilla and use the partial cyclicity of the partial trace to arrive at the following expression for the states $\ketbra{\varphi_i}{\varphi_i}$ 
\begin{equation} \label{eqn : ketbra varphii}
   \ketbra{\varphi_i}{\varphi_i} =  \frac{1}{q_i}\tTr_\cW\Bigg( \Big(I_\cH\otimes U^*\ketbra{i}{i}U\Big)\ketbra{\psi}{\psi} \Bigg)
\end{equation}
for each $i$ in $\{1,\dots,d\}$.

\begin{Con}
    Since all of the sums that follow range from 1 to $d$, we will often write $\sum_{index}$ to mean $\sum_{index=1}^d$. 
\end{Con}

The natural question that arises is how to search the unitary group for the unitary that minimizes the quantity $\sum_{i=1}^d q_i\mu(\varphi_i)$ where the $q_i$ and $\ket{\varphi_i}$ as in Equations \eqref{eqn : qi} and \eqref{eqn : ketbra varphii} respectively, and $\mu$ is an entanglement measure or related function. There are several ways to accomplish this task, the most common being the utilization of a parametrization of the unitary group \cite{CMB} and applying geometric optimization techniques as in \cite{roth}, or the use of random parametrized quantum circuits \cite{ibm1, Goog1, Goog2,wcubed}. We choose the latter, as this choice will greatly simplify calculations in the proof of our main result. Specifically, we will use parametrized unitaries of the following form:
\begin{equation}
\label{eqn : ansatz}
    U(\theta_1,\dots,\theta_L)=\prod_{l=1}^L exp\big(i\theta_l V_l\big)E_l,
\end{equation}
where $V_l$ is a Hermitian operator, $\theta_l\in[0,2\pi)$ for each $l\in\{1,\dots,L\}$, and $E_l$ is an entangling gate 
independent of the parameters $\theta_1,\dots,\theta_L$. Common choices for $E_l$ are a  $CX$ or $CZ$ ladder. In applications, 
the operators $V_l$ are often taken to be a randomly chosen string of Pauli operators, which is the convention that we use. Indeed, it is well known that the gates $R_x(\theta)= exp(-i \theta X/2)$, $R_y(\theta)= exp(-i \theta Y/2)$, $R_z(\theta)=exp(-i \theta Z/2)$, the phase gate $Ph(\theta)=exp(i\theta)$ and $CNOT$ form a universal gate set, (see \cite[Section 2.6]{Williams}).
Since we have introduced the parameters $\theta_1,\dots,\theta_L$ into our unitaries, the ansatz we wish to train can be written 
as $\sum_{i=1}^d q_i(\theta_1,\dots,\theta_L)\mu\Big(\varphi_i(\theta_1,\dots,\theta_L)\Big)$, thus allowing us to use common 
optimization techniques to approximate the minimum. Thus a quantum algorithm to approximate optimal pure state ensembles can be 
summarized as below:

\begin{algorithm}[H]
\caption{Variational Approximation of OPSE}
\begin{algorithmic}
\State  {\bfseries Input:} An ensemble $\{(p_i)_{i=1}^d, (\psi_i)_{i=1}^d\}$ of a density matrix $\rho\in\tD(\C^h)$, 
and an optimization routine, (as in \cite{nr}).
\State {\bfseries Parameters:} Initial parameters $\theta_1^0,\dots,\theta_L^0$.\\
\While{Optimization routine has not converged}{
    \State Prepare the purification $\ket\psi = \sum_{i=1}^d \sqrt{p_i}\ket{\psi_i}\otimes\ket{i}$ on $\cH\otimes\cW$.
    \State Construct $U(\theta_1,\dots,\theta_L)$ as above and apply $U(\theta_1,\dots,\theta_L)$ to the ancilla, $\cW$.
    \State Measure.
    \State Use measurements to compute $\sum_{i=1}^d q_i(\theta_1,\dots,\theta_L)\mu\Big(\varphi_i(\theta_1,\dots,\theta_L)\Big)$ (or related quantities such as gradients as required by the chosen optimization technique). 
    \State Update $\theta_1,\dots,\theta_L$ according to optimization routine.
}
\State {\bfseries return} $\bigg\{\Big(q_i(\theta_1,\dots,\theta_L)\Big)_{i=1}^d,\Big(\varphi_i(\theta_1,\dots,\theta_L)\Big)_{i=1}^d\bigg\}$.

\end{algorithmic}
\end{algorithm}


An algorithm of this type is called a \emph{variational quantum algorithm}(VQA) \cite{vqa1}. VQA has shown great success in various settings such as Hamiltonian Simulation and Quantum Machine Learning \cite{ibm1,vqa1, vqa2, vqa3} where most cost functions are physically implementable on a quantum computer. In the algorithm described above, we instead train a nonlinear cost function on the measurements of a physically implementable ansatz. While VQA is believed, (not yet proven), to offer a possible quantum advantage over classical algorithms for certain problems \cite[Chapter 1 - Section C]{vqa1}, there are still drawbacks, namely the issue of \emph{Barren Plateaus} \cite{bp1,bp2,bp3,bp4}. These are regions of the training landscape on which the partial derivatives with respect to the parameters decrease exponentially as the number of qubits in the experiment increase. More precisely, we give the following definition.

Consider the ansatz in Equation~(\ref{eqn : ansatz}) on $k$ qubits where $d=2^k$, with objective function 
\begin{equation}
C(\theta_1,\dots,\theta_L)=\sum_{i=1}^d q_i(\theta_1,\dots,\theta_L)\mu\Big(\varphi_i(\theta_1,\dots,\theta_L)\Big). 
\end{equation}

Then the training landscape will contain barren plateaus in the parameter $\theta_j$  for some $j \in \{1, \ldots , d \}$, if
\begin{equation}
\variance_{\U(2^k)}\big[\partial_{\theta_j} C(\theta_1,\dots,\theta_L)\big] \sim\mathscr{O}(a^{k})
\end{equation}
where the variance is taken with respect to the Haar measure on the Unitary group, $k$ is the number of qubits, and 
$a\in(0,1)$ is a real number \cite{bp3}. When this variance is exponentially small, the ansatz can become untrainable in practice because an exponentially large amount of precision will be required on the classical hardware to evaluate the derivatives. For instance, if it were shown that the variance is on the order of $\frac{1}{2^k}$, where $k$ is the number of qubits, then systems with only 100 qubits can start creating difficulties since the number of classical bits required to differentiate the values of the gradient from 0 will be in the order of $2^{-100}$.

In this article we show that the $f$-$d$-extension 
$T_{f,d}$ of the Tsallis entanglement entropy $T_2$ for $f(x)=x^2$  exhibits 
barren plateaus when using the ansatz given in Equation~\eqref{eqn : ansatz} 
of large depth, (thus in particular, if no additional information is 
available about the state whose entanglement is examined). In practice, if additional information about the state is known, then one needs to avoid using the suggested
ansatz for long depth of circuits. Recall that 
$T_{f,d}$ is given by Equation~\eqref{E:objective_function}.
While the OPSE of $T_{f,d}$ will yield $T_{f,d}(\rho)=0$ if we allow infinitely long decompositions as considered in 
Remark~\ref{infimum_is_zero}, $T_{f,d}$ is still useful in  
determining whether a state is entangled, as long as we consider decompositions of finite length, (by Proposition~\ref{prop:faithful}).
Also, obviously, $T_{f,d}$ is local unitary invariant (just as the Tsallis
entanglement entropy $T_2$), but is not LOCC monotonic.
In order to set the ground for Theorem~\ref{Thm:main_theorem}, first
define the cost function for the algorithm, that needs to be minimized, to be equal to $T_f(\rho;\theta_1,\dots,\theta_L)$ 
which is given by

\begin{equation}
T_{f,d}(\rho;\theta_1,\dots,\theta_L) = \sum_i q_i^2(\theta_1,\dots,\theta_L)T_2(\varphi_i(\theta_1,\dots,\theta_L)).
\end{equation}

To simplify notation somewhat, we define $\Phi_i$ by
\begin{equation}\label{eqn : bigphi}
\begin{split}
\Phi_i (\theta_1, \ldots , \theta_L)&=q_i(\theta_1, \ldots , \theta_L) \ketbra{\varphi_i(\theta_1,\ldots,\theta_L)}{\varphi_i(\theta_1,\ldots,\theta_L)}\\
&= \tTr_\cW\Bigg( \Big(I_\cH\otimes U^*\ketbra{i}U\Big)\ketbra{\psi}{\psi} \Bigg).
\end{split}
\end{equation} 
For readability, we will usually suppress the parameters 
$(\theta_1, \ldots , \theta_L)$ and simply write $\Phi_i$.
Since, by Equation~\eqref{Tsallis}, $T_{f,d}$ can be equivalently rewritten 
from Equation~\eqref{E:objective_function} to 

\begin{equation} \label{E:T_f_d}
T_{f,d} (\rho)= 
\inf 
\bigg\{ 
\sum_{i=1}^d q_i^2 ( 1 - \tTr ((\tTr_{\cA} (\ketbra{\psi_i}))^2)): 
\rho= \sum_{i=1}^d q_i \ketbra{\psi_i} \in \tD(\cH)_d 
\bigg\} ,
\end{equation}
we can write $T_{f,d}$ succinctly as 
\begin{equation}
    \sum_i\Big( \big(\tTr\tTr_\cA\Phi_i\big)^2 - \tTr\big(\tTr_\cA\Phi_i^2\big)\Big) = \sum_i\Big(q_i^2-\tTr\big(\tTr_\cA\Phi_i^2\big)\Big).
\end{equation}

 Notice that the last quantity is non-negative, since 
 $\Big(\tTr (P)\Big)^2 \geq \tTr( P^2)$ for all positive operators $P$. Indeed, let $P=\sum_i \lambda_i \Pi_i$ be the spectral decomposition of $P$. Then all eigenvalues $\lambda_i$ are non-negative and the $\Pi_i$ are mutually orthogonal projections with $\tTr \Pi_i=1$ for each $i$. Thus 
\begin{equation}
    \Big(\tTr(P)\Big)^2 = \sum_{i,j}\lambda_i\lambda_j
    = \sum_i \lambda_i^2 + \sum_{i\neq j}\lambda_i\lambda_j
    \geq \sum_i \lambda_i^2
    = \tTr(P^2).
\end{equation}

We show some modest limits on the trainability of  $T_{f}(\rho;\theta_1,\dots,\theta_L)$ as summarized by the following theorem.

\begin{Thm} \label{Thm:main_theorem}
    If the  parametrized circuit in Equation~\eqref{eqn : ansatz} is of sufficient depth, then for each $j \in \{1,\dots,L\}$, $\variance[\frac{\partial}{\partial_{\theta_j}} T_f]\sim\mathscr{O}(\frac{1}{2^k})$, where the objective function is defined in Equation~\eqref{E:objective_function}, the variance is taken with respect to the Haar measure on the unitary group, and $k$ is the number of qubits. 
\end{Thm}

Because of the length of the proof of the above theorem, its proof is presented in the appendix of this article. 
\begin{Rmk}
As shown in the proof of Theorem~\ref{Thm:main_theorem}, the circuit must have a relatively large depth in order for a barren plateau to occur. Thus for most practical purposes, the ansatz should still be trainable since less expressive circuits are often used in practice. It is also possible to construct more \lq\lq friendly\rq\rq ans\"atze by taking into consideration properties of the given density matrix such as translation invariance.
\end{Rmk}

\begin{figure}[h!]
  \captionsetup{width = 3.8 in}
  \includegraphics[width = 3.8 in]{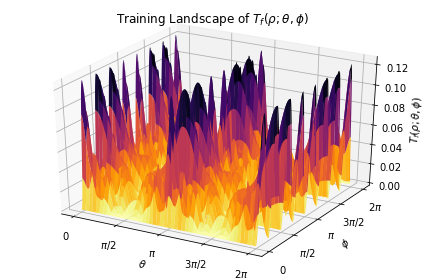}
  \caption{Training landscape of $T_f(\rho;\theta,\phi)$ where $\rho$ is the maximally mixed state and $\theta,\phi\in[0,2\pi]$.}
  \label{fig:tflandscape}
\end{figure}

\section{Numerical Simulations}
Recall that for a two qubit state $\ket\psi=\sum_{i,j=0}^1c_{i,j}\ket{ij}$, the \emph{concurrence} \cite{HHHH,conc} $\text{C}(\psi)$ is defined by
\begin{equation}
    \text{C}(\psi) = 2|c_{0,0}c_{1,1}-c_{0,1}c_{1,0}|.
\end{equation}
The concurrence is related to the von~Neumann entanglement entropy via the formula 
\begin{equation}
    S(\psi) = h\Big(\frac{1+\sqrt{1-C(\psi)^2}}{2}\Big)
\end{equation}
where $h$ is the standard binary entropy defined by $h(x)=-x\log_2(x)-(1-x)\log_2(1-x)$ \cite{HHHH}. Interestingly enough, it is possible to directly measure the concurrence of a 2 qubit quantum system if there is access to two decoupled copies of the same state \cite{conc}. One can do this using the quantum circuit in Figure~\ref{circ:2}, where the first two wires as well as  the last two wires are in the state $\ket\psi$ respectively, and $R$ is the unitary  
\begin{equation}R = \frac{1}{\sqrt{2}}
\begin{bmatrix}
    1 & 1 \\
    -1 & 1 \\

\end{bmatrix}.
\end{equation}

\begin{figure}[h!]
\centering
\captionsetup{width=3.8 in}
\includegraphics[width =3.8 in]{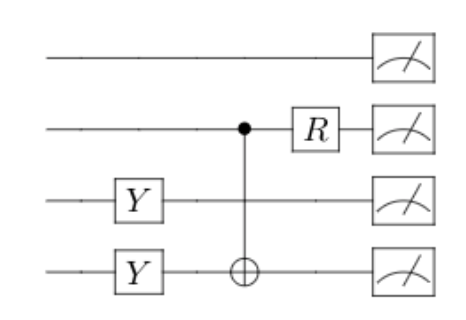}
\caption{Quantum circuit for directly measuring concurrence in 2 qubit states.}
\label{circ:2}
\end{figure}
\noindent After applying these operations to the state $\ket\psi\otimes\ket\psi$, the probability amplitude of the state $\ket{0000}$ will be $\pm\frac{C(\psi)}{2\sqrt{2}}$.

Using the above ideas, we can append the circuit for measuring the concurrence of a state to the pure state ensemble circuit, so that we can measure the concurrence of a 2 qubit ensemble of states directly, and therefore efficiently compute the entanglement of formation for 2 qubit states. The entire circuit for this process is given in Figure~\ref{fig : complicated} for an arbitrary density matrix $\rho$ with a purification $\ket\psi$ requiring two qubits. 

\begin{figure}[!ht]
\centering
\captionsetup{width=3.8 in}

\includegraphics[width = 3.8 in]{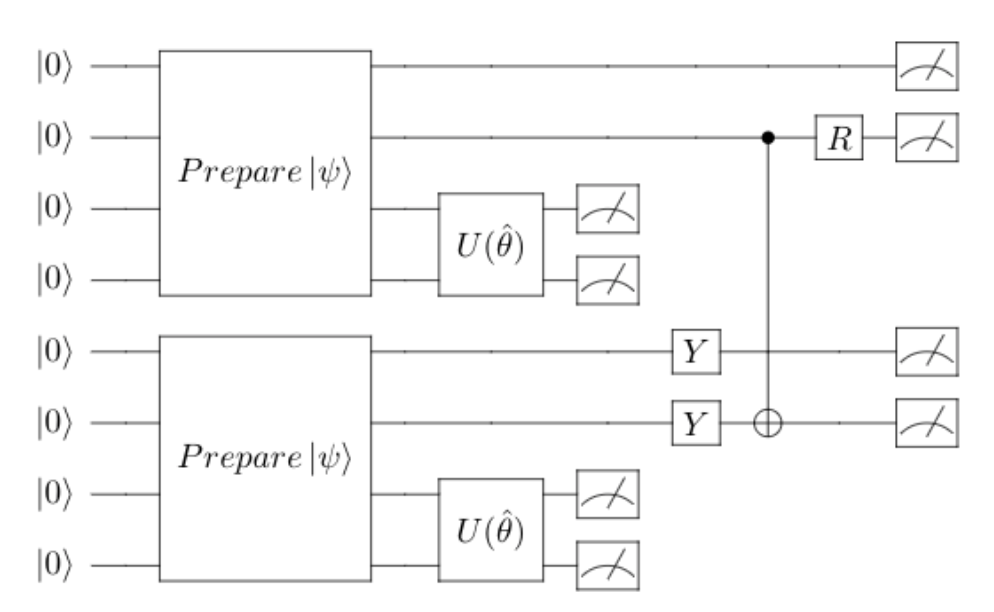}
\caption{Quantum circuit for variationally approximating the entanglement of formation by directly measuring the concurrence of pure state ensembles.}
\label{fig : complicated}
\end{figure}

Let's go through the circuit in Figure ~\ref{fig : complicated} step by step since there's a lot going on. We first prepare the purification $\ket\psi$ of a two qubit density operator on the first four and last four quantum registers so that we can access two decoupled copies of the pure state ensemble. Next we apply $U(\hat{\theta})$ to both of the ancilla spaces and measure the ancillae. This will put the top two registers into the state $\ket{\varphi_i}$, (as in Equation~\eqref{eqn : ketbra varphii}), the third and fourth registers into the state $\ket{i}$, the 5th and 6th registers into the state $\ket{\varphi_j}$, (again, as in Equation~\eqref{eqn : ketbra varphii}), and the last two registers into the state $\ket{j}$, where $i,j\in\{0,1,2,3\}$. Thus the state of circuit at this point is then $\ket{\varphi_i i \varphi_j j}$. If $i=j$, then we can measure the concurrence by appending the circuit in Figure ~\ref{circ:2} to wires 1,2,5, and 6, and taking note of the frequency of the state $\ket{00i00j}$ where the first two registers are in state $\ket{00}$, the third and fourth are in state $\ket{i}$, the fifth and sixth are in state $\ket{00}$, and lastly the 7th and 8th are in the state $\ket{j}$. We then run this experiment a desired number of times to also make note of the probabilities $q_i$ as defined as in Equation \eqref{eqn : qi} and the concurrence $C(\varphi_i)$ for each $i\in\{0,1,2,3\}$. Lastly we use the measurements to approximate the entanglement entropy in the equation below via 

\begin{equation}
    \sum_{i=0}^3 q_i h\Big(\frac{1+\sqrt{1-\text{C}^2(\varphi_i)}}{2}\Big) \approx \text{S}(\rho).
\end{equation}
Using the maximally mixed state and its canonical purification with 4 ancillae, we simulated the above variational algorithm and showcase our results in Figures~\ref{fig:straininglandscape} and \ref{fig:convergence}.
\begin{figure}
\captionsetup{width=3.8 in}
\includegraphics[width=3.8 in]{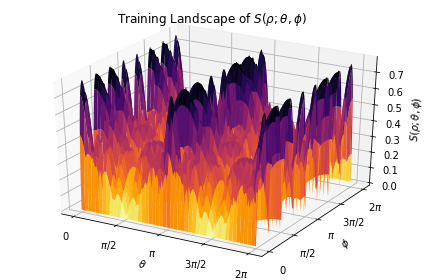}
  \captionof{figure}{The training landscape of $\text{S}(\rho;\theta,\phi)$ where $\rho$ is the maximally mixed state and $\theta,\phi\in[0,2\pi]$.}
  \label{fig:straininglandscape}
\end{figure}

\begin{figure}
  \centering
  \captionsetup{width=3.8 in}
  \includegraphics[width=3.8 in]{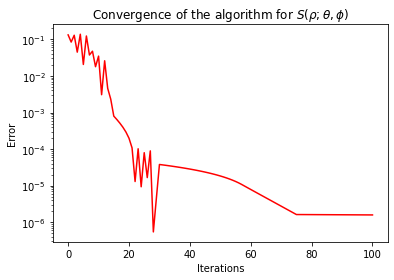}
  \captionof{figure}{Convergence of the algorithm for the entanglement of formation of the maximally mixed state.}
  \label{fig:convergence}
\end{figure}

The python code corresponding to the simulations and plots can be found at \url{https://github.com/TheMathDoctor/code_for_vqa_paper}. These simulations made extensive use of both the numpy \cite{numpy} and matplotplib \cite{matplotlib} python libraries.

\section{Approximating the Tsallis Entanglement Entropy Via The Swap Test} \label{sec:Swap_test}
In this section we discuss some additional applications of the algorithm discussed in Section 2. Specifically, we show that it is possible to measure the convex roof extended Tsallis entanglement entropy directly using the swap test \cite{swap}. 

Recall that the swap test can be used to estimate the quantity $|\bra{\psi}\ket{\phi}|^2$ with $\varepsilon$ error using the circuit in Figure ~\ref{fig : swap_test} using $\mathscr{O}\big(\frac{1}{\varepsilon^2}\big)$ quantum queries. When running this circuit a large number of times, an observer will measure $\ket{0}$ on the top wire with probability $\frac{1}{2}+\frac{1}{2}|\bra{\psi}\ket{\phi}|^2$, and an observer will measure $\ket{1}$ on the top wire with probability $\frac{1}{2}-\frac{1}{2}|\bra{\psi}\ket{\phi}|^2$.  

\begin{figure}[!ht]
\centering
\captionsetup{width=3.8 in}

\includegraphics[width = 3.8 in]{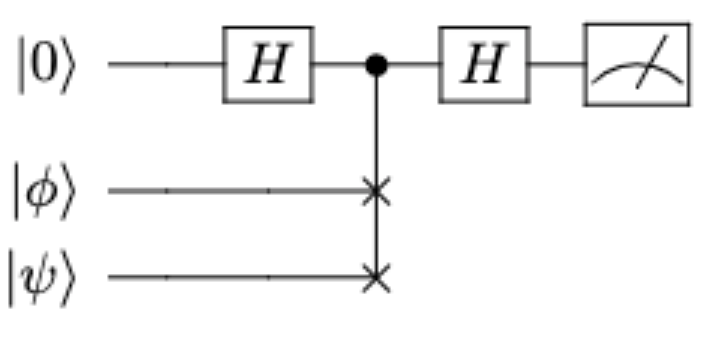}
\caption{Quantum circuit used for the swap test.}
\label{fig : swap_test}
\end{figure}

To see how this is useful for computing the Tsallis entropy of a state, consider a density operator $X=\sum_i \lambda_i \ketbra{v_i}$ where $\sum_i\lambda_i=1$ and the $\ket{v_i}$ are pure states which are not necessarily orthonormal. 
Then $\tTr(X^2)=\sum_{i,j}\lambda_i\lambda_j|\bra{v_i}\ket{v_j}|^2$. Thus, in order to compute the Tsallis entropy of $X$ which is given by $T(X) = 1 - \tTr(X^2)$, (not to be confused with the Tsallis entanglement entropy $T_2$ of a pure bipartite 
state given in Equation~\eqref{Tsallis}),
one needs to run the circuit in Figure ~\ref{fig : swap_test} with inputs $\ket{\phi}=\ket{v_i}$ and $\ket\psi = \ket{v_j}$ for each $i$ and $j$, measure the probability of measuring $\ket{1}$ for each combination, and then weight the outputs by the appropriate probabilities $\lambda_i\lambda_j$. In summary, if approximate values of $\lambda_i$'s can be computed,
then the Tsallis entropy $T(X)$  of $X$ can be approximated  via 
\begin{equation}
    T(X)= 2\sum_{i,j}\lambda_i\lambda_j \mathbb{P}\big[\text{top wire is }\ket{1}\big| \ket{\phi}=\ket{v_i},\ket{\psi}=\ket{v_j} \big].
\end{equation}

Now we need to see how circuit is useful for approximating the Tsallis entanglement entropy, let 
$\rho\in\tD\Big(\cA\otimes\cB\Big)$ and let $\ket\psi$ be a purification of $\rho$ as in Equation \eqref{eqn : psi}. Next 
let $\ket{\varphi_i}$ and $q_i$ be as in Equations \eqref{eqn : ketbra varphii} and \eqref{eqn : qi}. Then, the Tsallis 
entanglement entropy of $\rho$ can be approximated by minimizing the expression

\begin{equation}
    \sum_i q_i T\Big( \Tr_\cB\ketbra{\varphi_i}\Big) 
\end{equation}
with respect to the $q_i$ and $\ket{\varphi_i}$. However, the swap test requires two copies of the state, so we need to have two adjacent copies the pure state ensemble circuit, each of them as in Figure~\ref{circ:1}, in order to access two copies the various ensembles of $\rho$ and then append the quantum circuit for the swap 
test. Thus we need a quantum circuit very similar to the circuit in Figure~\ref{fig : complicated}. The full circuit 
for this process is given in figure~\ref{fig : tsallis}. \\

\begin{figure}[!ht]
\centering
\captionsetup{width=3.8 in}

\includegraphics[width = 3.8 in]{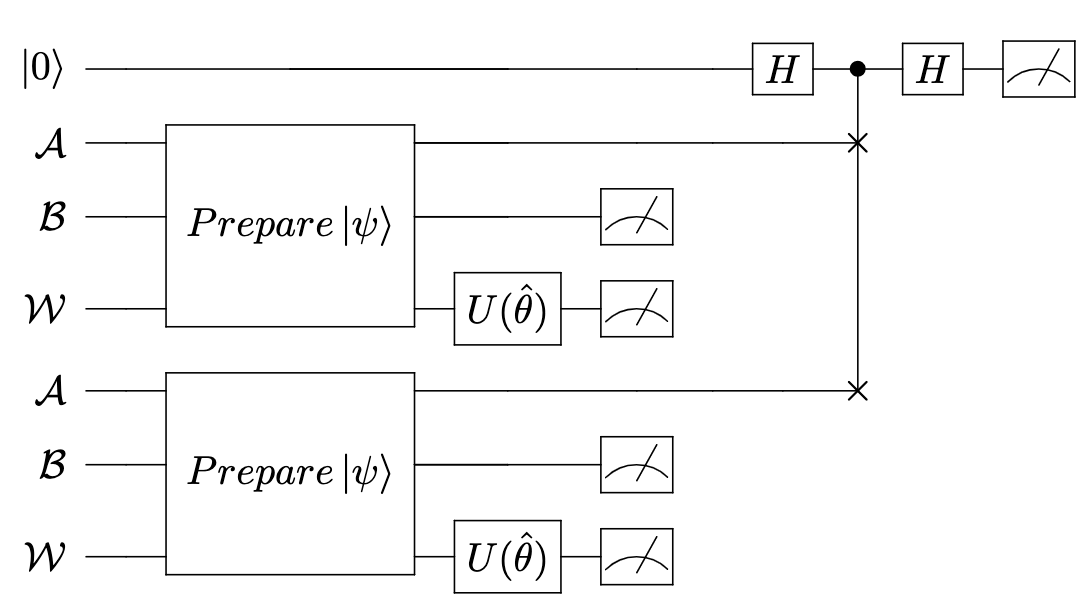}
\caption{Quantum circuit for approximating Tsallis Entanglement Entropy}
\label{fig : tsallis}
\end{figure}

Let's go through this circuit one operation at a time to understand what's going on, so that we can point out some loose 
quantum query complexity constraints on this process. Before we begin, notice that in this circuit, we specified wires for the Hilbert spaces $\cA$ and $\cB$ instead of specifying one wire for $\cH$. The reason for this is that we need to partial trace (measure and discard) either the $\cA$ or $\cB$ subsystems in order to approximate the Tsallis entanglement entropy. 

First, we need to prepare two decoupled copies of the purification $\ket\psi$ of a given density $\rho$ and apply 
$U(\hat{\theta})$ to each ancilla $\cW$. Then we measure each ancilla $\cW$. If we obtain measurements 
$i$ and $j$ on the two ancillas, then the state of the circuit will collapse into 
$\ket{0}\otimes\ket{\varphi_i}\otimes\ket{\varphi_j}$, (where $\ket{0}$ is on the top wire, and 
$\ket{\varphi_i}$ and $\ket{\varphi_j}$ are on the wires representing the two copies of $\cA \otimes \cB$).

If $i\neq j$, then all other operations will produce no useful information and we just use the observations to 
approximate $q_i$ and $q_j$. 
Otherwise, if $i=j$, then the state of the circuit will be
$\ket{0}\otimes\ket{\varphi_i}\otimes\ket{\varphi_i}$. In this case, we measure the space $\mathcal{B}$ with respect 
to a fixed basis (say the standard $Z$-basis) and collect all measurements weighted according to their probabilities 
into one density matrix on the space $\mathcal{A}$. 
The density matrix that is obtained that way is equal to 
$\tTr_\cB\ketbra{\varphi_i}=\sum_\alpha \lambda_{\alpha,i} \ketbra{\xi_{\alpha,i}}$ for some probabilities 
$\lambda_{\alpha,i}$ and some pure states $\ket{\xi_{\alpha,i}}$ on $\cA$ for each $i$.
Therefore, upon measuring and discarding the wire(s) for the space $\cB$, the first wire for $\cA$ will be in the state 
$\ket{\xi_{\alpha,i}}$, and the second wire for $\cA$ will be the state $\ket{\xi_{\alpha',i}}$ for some $\alpha$ and
$\alpha'$ each appearing with probabilities $\lambda_{\alpha,i}$ and $\lambda_{\alpha',i}$ respectively. 
Then, the swap test circuit given in Figure~\ref{fig : swap_test} is appended with inputs $\ket{\phi}= \ket{\xi_{\alpha,i}}$ and 
$\ket{\psi}=\ket{\xi_{\alpha',i}}$. As we noticed earlier, we will measure $\ket{1}$ on the top wire with probability 
$\frac{1}{2}- \frac{1}{2} |\bra{\xi_{\alpha,i}}\ket{\xi_{\alpha',i}}|^2$.
Moreover, the state  $\ket{\xi_{\alpha,i}}\otimes\ket{\xi_{\alpha',i}}$ will occur with frequency $\lambda_{\alpha,i}\lambda_{\alpha',i}N$ where $N$ is the number of shots. One can then run this experiment a desired number of times in order approximate $T(\tTr_\cB\ketbra{\varphi_i})$ which can seen to be equal to 
\begin{equation}
    T(\tTr_\cB\ketbra{\varphi_i}) = \sum_{\alpha,\alpha'}\lambda_{\alpha,i}\lambda_{\alpha',i}\big(1-|\bra{\xi_{\alpha,i}}\ket{\xi_{\alpha',i}}|^2\big)
\end{equation}
for each $i$. Thus, the Tsallis entanglement entropy of $\rho$ is approximated by 
\begin{equation}
    T_2(\rho)\approx \min_{\hat{\theta}}\sum_i q_i \sum_{\alpha,\alpha'}\lambda_{\alpha,i}\lambda_{\alpha',i}\big(1-|\bra{\xi_{\alpha,i}}\ket{\xi_{\alpha',i}}|^2\big).
\end{equation}

Next we need to think about the algorithmic complexity of this approximation. In order for the quantum circuit to yield useful information both 
ancillae must be in the same state $\ket{i}$ for some $i$. By equation \eqref{eqn : qi}, this happens with probability $q_i^2$. Thus if we consider 
the Bernoulli random variable $X$ such that $X=$\lq\lq i=j\rq\rq occurs with probability $q_i^2$ and $X=$\lq\lq i$\neq$j\rq\rq occurs with 
probability $1-q_i^2$, then the expected number shots until the first success is given by $\frac{1}{q_i^2}$. Since $\avg_\mathbb{U}[q_i]=\frac{1}{d}$ 
for each $i$, we can see that the average number of shots until the first success,
i.e.\ $\avg_\mathbb{U} [\frac{1}{q_i^2}]$, is then greater than or equal to $d^2$ by applying Jensen's inequality 
using the convex function $1/x^2$ for $x>0$. 
And since $d=2^k$ where $k$ is the number of qubits in the system, we can see that the approximation will require exponentially many shots on average in order to compute $ T(\tTr_\cB\ketbra{\varphi_i})$ within 
some error $\varepsilon$ for a fixed $\hat{\theta}$. Thus, while we are able to approximate the Tsallis entanglement entropy directly via the measurements of a quantum circuit, the process is inefficient.
While we do not prove it here, we suspect that there should be similar results for the von~Neumann entanglement entropy 
and related quantities using methods found in \cite{entcirc,newent}. We also leave as an open to problem to determine 
if this objective function suffers from barren plateaus.   

\section{Concluding Remarks}
We have shown how to access the different pure state ensembles of a given density operator using a quantum computer. We are thus able to approximate entanglement measures and related quantities that use convex roof constructions using a variational method and measurements from a parametrized quantum circuit. However we also show that the training landscape of a quantity related to the Tsallis entanglement entropy exhibits barren plateaus for a certain ansatz
and for circuits of large depth, effectively making the ansatz untrainable for large enough systems. In practice, if additional information about the state is known, then one needs to avoid using the suggested
ansatz for long depth of circuits. Lastly, we leave as an open problem the determination of barren plateaus of other entanglement measures and related quantities.

\bibliographystyle{unsrt}
\bibliography{bibliography}

\newpage
\appendix
\section{Integration Over The Unitary Group}
In order to prove Theorem~\ref{Thm:main_theorem}, we'll need the following results on integrals over the unitary group \cite{haar,bc}:
\begin{Lem}\label{eqn : wein}
    For any strings of indices $i_1,\dots,i_p,j_1,\dots,j_p$ and $i_1',\dots,i_q',j_1',\dots,j_q'$ whose values range from $1$ to $d$ 
    \begin{equation}
        \int_{\cU(d)}U_{i_1,j_1}\cdot\cdot\cdot U_{i_p,j_p}\overline{U_{i_1',j_1'}}\cdot\cdot\cdot\overline{U_{i_q',j_q'}}dU=\delta_{p,q} \sum_{\sigma,\tau\in \text{S}_p}\text{Wg}(\sigma\tau^{-1})\delta_{i_1,i_{\sigma(1)}'}\delta_{j_1,j_{\tau(1)}'}\cdot\cdot\cdot\delta_{i_p,i_{\sigma(p)}'}\delta_{j_p,j_{\tau(p)}'}
    \end{equation}
    where $dU$ is the Haar-Measure on the Unitary group and $\text{Wg}$ is the Weingarten function \cite{wg1,wg2}.
\end{Lem}
While notation can make this lemma difficult to parse at first glance, notice that the integral is 0 when the number of complex conjugate entries of the matrix is different the number of non-conjugated entries. Unfortunately, the integrals used in this paper aren't lucky enough to fall into this category and we must instead use the expression on the right hand side of the equation to evaluate our integrals. Since this expression can be difficult to understand at first, we supply a simple example of computing such an integral via this lemma. 

\begin{Ex}
    Let $U,V$ be unitary matrices on $\C^d$. Then the average gate fidelity between $U$ and $V$, $F_{avg}(U,V)$, is defined by 
    \begin{equation}
        F_{avg}(U,V) := \int_{\text{S}(\cH)} |\bra\psi U^* V\ket\psi |^2 d\psi
    \end{equation}
    where $\text{S}(\cH)$ is unit ball of $\cH$.
\end{Ex} 
Such an integral over pure states is defined by 
\begin{equation}\label{eqn : intstates}
    \int_{\text{S}(\cH)} |\bra\psi U^* V\ket\psi |^2 d\psi := \int_{\cU(d)}|\bra{\psi_0}W^* U^* VW\ket{\psi_0}|^2 dW
\end{equation}
where each $\psi$ is replaced by $W\ket{\psi_0}$ for some arbitrary state $\ket{\psi_0}\in\cH$. Now let $\{\ket{i} : 1\leq i \leq d\}$ be a basis for $\C^d$ and without loss of generality, set $\ket{\psi_0}:=\ket{1}$. Next write $H = U^*V$ and suppose that $H$ has spectral decomposition 
\begin{equation}
    H = \sum_i \sigma_i P^*\ketbra{i}P
\end{equation}
where $\sigma_i\in\C$ with $|\sigma_i|=1$ for each $1\leq i \leq d$ and $P$ is some unitary on $\cH$. Using these substitutions, the integral in Equation\eqref{eqn : intstates} transforms as 
\begin{equation}
\begin{split}
    &\int_{\cU(d)}|\bra{\psi_0}W^* U^* VW\ket{\psi_0}|^2 dW\\
    &= \int_{\cU(d)}\bra{1}W^* HW\ket{1}\overline{\bra{1}W^* HW\ket{1}} dW\\
    &=\sum_{i,j}\sigma_i\overline{\sigma}_j\int_{\cU(d)}\bra{1}W^* P^*\ketbra{i}PW\ket{1}\overline{\bra{1}W^* P^*\ketbra{i} PW\ket{1}} dW.\\
\end{split}
\end{equation}
Now using the translation invariance of the Haar measure on $\cU(d)$, we use a change of variable to simplify the above integral. Set $Q=PW$ so that $dQ=d(PW)=dW$. Therefore 
\begin{equation}
\begin{split}
    &\sum_{i,j}\sigma_i\overline{\sigma_j}\int_{\cU(d)}\bra{1}W^* P^*\ketbra{i}PW\ket{1}\overline{\bra{1}W^* P^*\ketbra{i} PW\ket{1}} dW\\ 
    &= \sum_{i,j}\sigma_i\overline{\sigma_j}\int_{\cU(d)}\bra{1}Q^*\ketbra{i}Q\ket{1}\overline{\bra{1}Q^*\ketbra{i} Q\ket{1}} dQ\\
    &= \sum_{i,j}\sigma_i\overline{\sigma_j}\int_{\cU(d)}Q_{i,1}Q_{j,1}\overline{Q_{i,1}}\text{ }\overline{Q_{j,1}} dQ. \\
\end{split}
\end{equation}
We can now use the integration formula given in \eqref{eqn : wein} to evaluate the integral in the last expression of the previous equation and we arrive at
\begin{equation}
\begin{split}
    &\sum_{i,j}\sigma_i\overline{\sigma_j}\int_{\cU(d)}Q_{i,1}Q_{j,1}\overline{Q_{i,1}}\text{ }\overline{Q_{j,1}} dQ \\ 
    &= \sum_{i,j}\sigma_i \overline{\sigma_j} \Big( \text{Wg}((1))(1 + \delta_{i,j}) + \text{Wg}((1,2))(1 + \delta_{i,j})\Big)\\
    &= \frac{1}{d(d+1)}\sum_{i,j}\sigma_i \overline{\sigma_j} (1 + \delta_{i,j}).
\end{split}
\end{equation}
where $\text{Wg}\big((1)\big)=\frac{1}{d^2-1}$ and $\text{Wg}\big((1,2)\big) = \frac{-1}{d(d^2-1)}$. Since $H$ is a unitary, each $\sigma_i$ has absolute value 1, and so 
\begin{equation}
    \sum_{i,j}\sigma_i\overline{\sigma_j}\delta_{i,j} = \sum_{i}|\sigma_i|^2 = d.
\end{equation}
To compute the other sum, we use the fact that the sum of the eignvalues of a matrix is its trace, so that
\begin{equation}
\begin{split}
    \sum_{i,j}\sigma_i\overline{\sigma_j} =\Big(\sum_i \sigma_i\Big)\Big(\sum_j \overline{\sigma_j}\Big) =  |\tTr(H)|^2.
\end{split}
\end{equation}
Recalling that $H=U^* V$, we can see that $\tTr(H)=\langle U,V\rangle_{HS}$ where the inner product is taken to be the Hilbert-Schmidt inner product. 
Thus it follows that 
\begin{equation}
    F_{avg}(U,V) = \frac{|\langle U, V\rangle|^2 + d}{d(d+1)}.
\end{equation}
 We can see that in order to compute these integrals, we must be familiar with the Weingarten function on the unitary group. In general, the Weingarten function can be difficult to compute, and we therefore include a table of values for the Weingarten function for the first four permutation groups as we'll need these values in the computation of the variance of concern. In table \eqref{tab:wg}, we list the cycle structure of a given permutation $\sigma$ in the left hand column and the the value of $\text{Wg}(\sigma)$ in the right hand column. These values were computed using the tool \cite{tool} written by M. Fukuda et alii.

\begin{table}[h!]
\captionsetup{width=3.8 in}
\caption{Table of Values for the Weingarten function}
\centering

\begin{tabular}{|c|c|}
\hline \hline
Cycle Structure of $\sigma$ & $\text{Wg}(\sigma)$ \\
\hline
   1 & $\frac{1}{d}$ \\
    \hline
    1,1 & $\frac{1}{d^2-1}$ \\
    \hline
    2 &  $\frac{-1}{d(d^2-1)}$\\
    \hline
    1,1,1 & $\frac{d^2-2}{d(d^4-5d^2+4)}$ \\
    \hline
    2,1 &  $\frac{-1}{d^4-5d^2+4}$\\
    \hline
    3 & $\frac{2}{d(d^4-5d^2+4)}$\\
    \hline     
    1,1,1,1 & $ \frac{d^4-8d^2+6}{d^2(d^6-14d^4+49d^2-36)}$ \\
    \hline
    2,1,1 &  $\frac{-1}{d(d^4-10d^2+9)}$\\
    \hline
    2,2& $\frac{d^2+6}{d^2(d^6-14d^4+49d^2-36)}$\\
    \hline 
    3,1 &  $\frac{2d^2-3}{d^2(d^6-14d^4+49d^2-36)}$\\
    \hline
    4 & $\frac{-5}{d(d^6-14d^4+49d^2-36)}$\\
    \hline
\end{tabular}
\label{tab:wg}
\end{table}

Notice that the value of the Weingarten function of a given permutation is dependent only on the cycle structure of a permutation. For example, let $\sigma = (1 2) (34)\in\text{S}_4$. This permutation has cycle structure 2,2 and so $\text{Wg}(\sigma)=\frac{d^2+6}{d^2(d^6-14d^4+49d^2-36)}$.


\section{Proof of Theorem~\ref{Thm:main_theorem}}
Before computing the mean and variance of the derivatives of the cost function $T_f(\rho;\theta_1,\dots,\theta_L)$, we must first discuss the probability distribution of unitaries produced by the circuits defined in Equation\eqref{eqn : ansatz}. 

In order to make the notation more readable, suppose that $\theta = \theta_j$ for some $j\in\{1,\dots,L\}$ and that $V=V_j$. Then the derivative of $U$ with respect to $\theta$ can be written as 
\begin{equation}
    \dt U(\theta) = i\Big(\prod_{l_1=1}^{j}\exp(i\theta_{l_1}V_{l_1})E_{l_1}\Big)V\Big(\prod_{l_2=j+1}^L\exp(i\theta_{l_2}V_{l_2})E_{l_2}\Big).
\end{equation}
The derivative splits the circuit into two \lq\lq pieces", the left side and the right side which we define accordingly 
\begin{equation} \label{E:L_and_R}
    L = \prod_{l_1=1}^{j}\exp(i\theta_{l_1}V_{l_1})E_{l_1} \text{ and } R=\prod_{l_2=j+1}^L\exp(i\theta_{l_2}V_{l_2})E_{l_2}
\end{equation}
which lets us write $\partial_\theta U(\theta)=iLVR$. Moreover, using this expression for the derivative of $U$ lets us express $\partial_\theta U^*\ketbra{i}U$ as 
\begin{equation}\label{eqn : commutator}
    \dt U^*\ketbra{i}U = \sqrt{-1}R^*\big[V,L^*\ketbra{i}L\big]R.
\end{equation}

We will often write $K_i= \sqrt{-1}\big[V,L^*\ketbra{i}L\big]$ for readability so that 
\begin{equation}
     \dt U^*\ketbra{i}U = R^*K_iR.
\end{equation}

This splitting forces us to consider the distributions of the $L$ and $R$ separately. Therefore, just as in the paper \cite{Goog2} by McClean et alii, we write the probability distribution generated by the entire circuit as

\begin{equation}\label{eqn : distribution}
    \rho(U)dU = \delta(U-LR)\rho_1(L)\rho_2(R)dLdR
\end{equation}
where $dU,dL,$ and $dR$ are the Haar measure on $\cU(d)$, $\rho_1$ and $\rho_2$ are densities on $\cU(d)$, and $\delta$ is the dirac measure centered at the $0-$matrix. This distribution forces any unitary $U$ to be split as the product of two other unitaries $L$ and $R$ with their own distributions $\rho_1$ and $\rho_2$ respectively. To quantify these distribution we'll need to use the notion of \emph{unitary k-designs}. A distribution $P$ is defined to be a k-design if $P$ matches the Haar measure on the unitary group up to and including the k-th moment \cite{2des,kdes,haar,bc}. In other words, 
\begin{equation}
    \int_{\cU(d)}U^{\otimes l}\otimes (U^*)^{\otimes l}P(U)d\mu(U) = \int_{\cU(d)}U^{\otimes l}\otimes (U^*)^{\otimes l} d\mu(U)
\end{equation}
for all $0\leq l \leq k$. In practice, exact $k$-designs are costly to produce. And so we'll also need the notion of an $\varepsilon-$approximate $k$-design \cite{kdes}. We say that a a measure $\nu$ is an $\varepsilon$-approximate $k$-design if and only if 
\begin{equation}
    (1-\varepsilon)\int_{\cU(d)}U^{\otimes k} X (U^*)^{\otimes k}dU \leq \int_{\cU(d)}U^{\otimes k}X (U^*)^{\otimes k}d\nu(U) \leq (1+\varepsilon)\int_{\cU(d)}U^{\otimes k}X (U^*)^{\otimes k}dU
\end{equation}
for any $X\in \tL\big(\C^{d\otimes k}\big)$, where $\varepsilon\in(0,1)$ and $dU$ is the Haar measure on $\cU(d)$. Such measures $\nu$ are attainable in practice using only polynomially many gates in the number of qubits \cite{2des, kdes}, thus the result is relevant for NISQ devices.

With the distribution $\rho$ from Equation \eqref{eqn : distribution} in mind, we will split the computation of the mean and variance of $\partial_\theta T_f(\rho;\theta)$ into two cases; One in which we assume that $\rho_1$ is a 4-design, and the other in which $\rho_2$ is a 4-design. Note that this assumption is a practical one since it has been proven \cite{2des, kdes} that polynomial depth quantum circuits such as those given in Equation \eqref{eqn : ansatz} are $\epsilon-$approximate k-designs, if they are of sufficient depth. Therefore, since the gates used the ansatz of 
Equation~\eqref{eqn : ansatz} form a universal gate set, if the depth of that
ansatz is sufficiently large, then at least one of $L$ and $R$ that appear in Equation~\eqref{E:L_and_R} has sufficient depth in order to yield that $\rho_1$ or $\rho_2$ is a 4-design.

Next let's expand the various components of $\dt T_f(\rho;\theta)$ using the definitions of $q_i$ and $\Phi_i$, given in Equations \eqref{eqn : qi} and \eqref{eqn : bigphi} respectively, for each $i \in \{1,\dots,d\}$.
\begin{equation}\label{eqn : qi2}
\begin{split}
    q_i &= \bra{\psi}\Big(I_{\cA\cB}\otimes U^*\ketbra{i}U\Big)\ket{\psi}\\
    &= \Big(\sum_{j_1'}\sqrt{p_{j_1'}}\bra{\psi_{j_1'}}\otimes\bra{j_1'}\Big)\Big(I_{\cA\cB}\otimes  U^*\ketbra{i}U\Big)\Big(\sum_{j_1}\sqrt{p_{j_1}}\ket{\psi_{j_1}}\otimes\ket{j_1}\Big)\\
    &= \sum_{j_1,j_1'}\sqrt{p_{j_1}p_{j_1'}}\bra{\psi_{j_1'}}\ket{\psi_{j_1}}\bra{j_1'}U^*\ketbra{i}U\ket{j_1}\\
    &= \sum_{j_1,j_1'}\sqrt{p_{j_1}p_{j_1'}}\bra{\psi_{j_1'}}\ket{\psi_{j_1}}\bra{i}U\ket{j_1}\overline{\bra{i}U\ket{j_1'}}\\
    &= \sum_{j_1,j_1'}\sqrt{p_{j_1}p_{j_1'}}\bra{\psi_{j_1'}}\ket{\psi_{j_1}}\bra{i}LR\ket{j_1}\overline{\bra{i}LR\ket{j_1'}}.\\
\end{split}
\end{equation}

It can be shown in a similar fashion that 
\begin{equation}
    \tTr_\cA\Phi_i = \sum_{j_1,j_1'}\sqrt{p_{j_1}p_{j_1'}}\text{ }\tTr_\cA\ketbra{\psi_{j_1}}{\psi_{j_1'}}\bra{i}LR\ket{j_1}\overline{\bra{i}LR\ket{j_1'}}.\\
\end{equation}

Using the commutator expression given in Equation \eqref{eqn : commutator} for the derivative of $q_i$ as follows: 
\begin{equation}
\begin{split}
    \frac{\partial q_i}{\partial\theta} &= \bra\psi \Big( I_{\cA\cB}\otimes \dt U^*\ketbra{i}U\Big)\ket\psi\\
    &= \bra\psi \Big( I_{\cA\cB}\otimes R^*K_iR\Big)\ket\psi\\
    &= \sum_{j_1,j_1'}\sqrt{p_{j_1}p_{j_1'}}\bra{\psi_{j_1'}}\ket{\psi_{j_1}}\bra{j_1'} R^*K_i R \ket{j_1}.\\
\end{split}
\end{equation}

In the same way, a similar expression is found for $\dt\tTr_\cA\Phi_i$:
\begin{equation}\label{eqn : dphii}
   \frac{\partial \tTr_\cA\Phi_i}{\partial\theta} = \sum_{j_1,j_1'}\sqrt{p_{j_1}p_{j_1'}}\text{ }\tTr_\cA\ketbra{\psi_{j_1}}{\psi_{j_1'}}\bra{j_1'} R^*K_i R \ket{j_1}.
\end{equation}
Putting all of this together we see that 
\begin{equation}
\begin{split}\label{eqn : dqsquared}
    \dt q_i^2 &= 2q_i\frac{\partial q_i}{\partial\theta}\\
    &= \sum_{j_1,j_2,j_1',j_2'}\sqrt{p_{j_1}p_{j_1'}p_{j_2}p_{j_2'}}\bra{\psi_{j_1'}}\ket{\psi_{j_1}}\bra{\psi_{j_2'}}\ket{\psi_{j_2}}\bra{i}LR\ket{j_1}\overline{\bra{i}LR\ket{j_1'}}\bra{j_2'} R^*K_i R \ket{j_2}.\\
\end{split}
\end{equation}

Also, $\dt\tTr\Big(\tTr_\cA\Phi_i^2\Big)$ has a similar expression:

\begin{equation}
    \dt\tTr\Big(\tTr_\cA\Phi_i^2\Big) = 2\tTr\Big(\tTr_\cA\Phi_i \frac{\partial \tTr_\cA\Phi_i}{\partial\theta}\Big)\\
\end{equation}

which can be seen to be equal to 
\begin{footnotesize}
\begin{equation}\label{eqn : dphisquared}
    2\sum_{j_1,j_2,j_1',j_2'}\sqrt{p_{j_1}p_{j_1'}p_{j_2}p_{j_2'}}\text{ }\tTr\Big(\tTr_\cA\ketbra{\psi_{j_1}}{\psi_{j_1'}}\tTr_\cA\ketbra{\psi_{j_2}}{\psi_{j_2'}}\Big)\bra{i}LR\ket{j_1}\overline{\bra{i}LR\ket{j_1'}}\bra{j_2'} R^*K_i R \ket{j_2}.
\end{equation}
\end{footnotesize}

Putting all of this together, we can expand $\dt T_f(\rho;\theta)$ as

\begin{equation}\label{eqn : deriv}
    \dt T_f(\rho;\theta) = 2\sum_i \Bigg( q_i \frac{\partial q_i}{\partial\theta}-\tTr\Big(\tTr_\cA\Phi_i \frac{\partial \tTr_\cA \Phi_i}{\partial \theta}\Big)\Bigg),
\end{equation}
and then integrate this expression term by term using the expressions above. In what follows, we'll fix an $i$ and integrate the terms $q_i \frac{\partial q_i}{\partial\theta}$ and $\tTr\Big(\tTr_\cA\Phi_i \frac{\partial \tTr_\cA \Phi_i}{\partial \theta}\Big)$ separately for clarity. After the computation of the integrals, we'll then sum the results to arrive at the appropriate expected value.

Now notice that the expressions given in equations \eqref{eqn : dqsquared} and \eqref{eqn : dphisquared}  are of the form 
\begin{equation} \label{eqn : sumc}
    2\sum_{j_1,j_2,j_1',j_2'}\sqrt{p_{j_1}p_{j_1'}p_{j_2}p_{j_2'}}c_{j_1,j_1',j_2,j_2'}\bra{i}LR\ket{j_1}\overline{\bra{i}LR\ket{j_1'}}\bra{j_2'} R^*K_i R \ket{j_2}
\end{equation}
where $|c_{j_1,j_1',j_2,j_2'}|\leq1$ for any choice of indices. Because we only need an upper estimate for our variance, we'll soon see that we don't need the exact values of the $c_{j_1,j_1',j_2,j_2'}$ and that we'll only use the fact that they are bounded in absolute value by 1. Next we need to integrate the above expression with respect to the distribution $\rho_1(L)\rho_2(R)dL dR$. Since only the expression $\bra{i}LR\ket{j_1}\overline{\bra{i}LR\ket{j_1'}}\bra{j_2'} R^*K_i R \ket{j_2}$ is dependent on $L$ and $R$, we can integrate this expression and substitute its value back into Equation \eqref{eqn : sumc}. 

In the integration formula given in Lemma \ref{eqn : wein}, the expressions that were integrated were given in terms of the coordinates of the unitaries, but the expression we currently have is not in this form. We must therefore force the expression into this form and expand the commutator $R^*K_iR$. It's easy to see that $\bra{i}LR\ket{j_1}\overline{\bra{i}LR\ket{j_1'}}\bra{j_2'} R^*K_i R \ket{j_2}$ is, up to a multiple of $\sqrt{-1}$, equal to 
\begin{equation}
    \bra{i}LR\ket{j_1}\overline{\bra{i}LR\ket{j_1'}}\cdot\Big(\bra{j_2'}R^*VL^*\ketbra{i}LR\ket{j_2}-\bra{j_2'}R^*L^*\ketbra{i}LVR\ket{j_2}\Big).
\end{equation}
For our computations, we'll just use the first term, $\bra{j_2'}R^*VL^*\ketbra{i}LR\ket{j_2}$, of the commutator since the computations with the second term are nearly identical. Suppose for now that $\rho_1(L)dL$ is at least a 2-design, then $\rho_1(L)dL$ is equal to the Haar-distribution $dL$ up to the second moment. Thus 
\begin{equation}
\begin{split}
    &\int\int\bra{i}LR\ket{j_1}\overline{\bra{i}LR\ket{j_1'}}\bra{j_2'}R^*VL^*\ketbra{i}LR\ket{j_2}\rho_1(L)\rho_2(R)dLdR \\
    &= \int\int\bra{i}LR\ket{j_1}\bra{i}LR\ket{j_2}\overline{\bra{i}LR\ket{j_1'}}\text{ }\overline{\bra{i}LVR\ket{j_2'}}\rho_2(R)dLdR.
\end{split}
\end{equation}
Using the translation invariance of the Haar-integral on the Unitary group, we define the change of variables $W=LR$ with $dW=dL$. The integral then becomes 
\begin{equation}
\begin{split}
    &\int\int\bra{i}LR\ket{j_1}\bra{i}LR\ket{j_2}\overline{\bra{i}LR\ket{j_1'}}\text{ }\overline{\bra{i}LVR\ket{j_2'}}\rho_2(R)dLdR\\
    &=\int\int\bra{i}W\ket{j_1}\bra{i}W\ket{j_2}\overline{\bra{i}W\ket{j_1'}}\text{ }\overline{\bra{i}WR^*VR\ket{j_2'}}dW\rho_2(R)dR\\
    &= \int\int W_{i,j_1}W_{i,j_2}\overline{W_{i,j_1'}}\text{ }\overline{\bra{i}WR^*VR\ket{j_2'}}dW\rho_2(R)dR.\\
\end{split}
\end{equation}

The last thing that we need to do before using the integration formula is to express $\overline{\bra{i}WR^*VR\ket{j_2'}}$ in terms of the coordinates of $W$. To do this, first define $Q = R^*VR$ so that 
\begin{equation}\label{eqn : q1}
\begin{split}
    R^*VR\ket{j_2'} & = Q\ket{j_2'}\\
    &= \sum_\alpha Q_{\alpha,j_2'}\ket{\alpha},\\
\end{split}
\end{equation}

and thus 

\begin{equation}\label{eqn : q2}
\begin{split}
    \overline{\bra{i}WR^*VR\ket{j_2'}} &= \overline{\bra{i}WQ\ket{j_2'}}\\
    &= \sum_\alpha \overline{Q_{\alpha,j_2'}}\text{ } \overline{\bra{i}W\ket{\alpha}}.\\
\end{split}
\end{equation}
Substituting this expression back into the integral we see that 
\begin{equation}\label{eqn : doubleint1}
\begin{split}
    &\int\int W_{i,j_1}W_{i,j_2}\overline{W_{i,j_1'}}\text{ }\overline{\bra{i}WR^*VR\ket{j_2'}}dW\rho_2(R)dR\\
    &= \sum_\alpha \int \overline{Q_{\alpha,j_2'}}\int W_{i,j_1}W_{i,j_2}\overline{W_{i,j_1'}}\text{ }\overline{\bra{i}W\ket{\alpha}}dW\rho_2(R)dR\\
    &= \sum_\alpha \int \overline{Q_{\alpha,j_2'}}\int W_{i,j_1}W_{i,j_2}\overline{W_{i,j_1'}}\text{ }\overline{W_{i,\alpha}}dW\rho_2(R)dR.\\
\end{split}
\end{equation}
Notice that all of the row coordinates of the $W$'s in the integral are equal to $i$, so we can suppress the $\delta_{i_k,i_{\sigma(k)}'}$ terms in the integration formula from Lemma \eqref{eqn : wein} since they are always equal to 1. The index $\alpha$ also took the place of $j_2'$ in the column coordinate of one of the $W$s in the last expression of the above equation.Therefore the $\alpha$ must be treated as a $j_{k}'$ when using the integration formula from Lemma~\ref{eqn : wein}. For clarity, we isolate the inner integral and compute it as follows:
\begin{equation}
\begin{split}
    &\int W_{i,j_1}W_{i,j_2}\overline{W_{i,j_1'}}\text{ }\overline{W_{i,\alpha}}dW\\
    &= \text{Wg}\Big((1)_2\Big)\delta_{j_1,j_1'}\delta_{j_2,\alpha} \hspace{.5 in} (\sigma=\tau = (1)_2 )\\
    &+ \text{Wg}\Big((1,2)_2\Big) \delta_{j_1,j_1'}\delta_{j_2,\alpha} \hspace{.5 in} (\sigma=(1,2)_2,\tau = (1)_2 )\\
    &+ \text{Wg}\Big((1,2)_2\Big)\delta_{j_1,\alpha}\delta_{j_2,j_1'} \hspace{.5 in} (\sigma=(1)_2,\tau = (1,2)_2 )\\
    &+ \text{Wg}\Big((1)_2\Big)\delta_{j_1,\alpha}\delta_{j_2,j_1'} \hspace{.5 in} (\sigma=\tau = (1,2)_2 )\\
    &= \Big(\text{Wg}(1,1) + \text{Wg}(2)\Big)\Big(\delta_{j_1,j_1'}\delta_{j_2,\alpha} + \delta_{j_1,\alpha}\delta_{j_2,j_1'}\Big).\\
\end{split}
\end{equation}
Defining $C:=\Big(\text{Wg}(1,1) + \text{Wg}(2)\Big)=\frac{1}{d(d+1)}$ and substituting the last expression of the above equation back into Equation \eqref{eqn : doubleint1}, we see that the value of the integral is equal to 

\begin{equation}
\begin{split}
    &C\sum_\alpha \int \overline{Q_{\alpha,j_2'}}\Big(\delta_{j_1,j_1'}\delta_{j_2,\alpha} + \delta_{j_1,\alpha}\delta_{j_2,j_1'} \Big)\rho_2(R)dR\\
    &=C \int \overline{Q_{j_2j_2'}}\delta_{j_1,j_1'}\rho_2(R)dR +C\int  \overline{Q_{j_1,j_2'}}\delta_{j_2,j_1'}\rho_2(R)dR.
\end{split}
\end{equation}
Ignoring the integral and outside constants for the moment, let's plug the expressions $\overline{Q_{j_2j_2'}}\delta_{j_1,j_1'}$ and $\overline{Q_{j_1,j_2'}}\delta_{j_2,j_1'}$ into Equation \eqref{eqn : sumc} and try to bound the sum in absolute value. Substituting $\overline{Q_{j_2,j_2'}}\delta_{j_1,j_1'}$, we see that
\begin{equation}
    \begin{split}
        \sum_{j_1,j_2,j_1',j_2'}\sqrt{p_{j_1}p_{j_1'}p_{j_2}p_{j_2'}}c_{j_1,j_1',j_2,j_2'}\overline{Q_{j_2,j_2'}}\delta_{j_1,j_1'} & = \sum_{j_1,j_2,j_2'}p_{j_1}\sqrt{p_{j_2}p_{j_2'}}c_{j_1,j_1',j_2,j_2'}\overline{Q_{j_2,j_2'}}.\\
    \end{split}
\end{equation}
Now taking absolute values, 
\begin{equation}
\begin{split}
\Big|\sum_{j_1,j_2,j_2'}p_{j_1}\sqrt{p_{j_2}p_{j_2'}}c_{j_1,j_1',j_2,j_2'}\overline{Q_{j_2,j_2'}}\Big|&\leq \sum_{j_1,j_2,j_2'}p_{j_1}\sqrt{p_{j_2}p_{j_2'}}\big|c_{j_1,j_1',j_2,j_2'}\big|\cdot \big|Q_{j_2,j_2'}\big| \\
&\leq \sum_{j_1,j_2,j_2'}p_{j_1}\sqrt{p_{j_2}p_{j_2'}} \big|Q_{j_2,j_2'}\big|\\
&= \Bigg(\sum_{j_1}p_{j_1}\Bigg)\cdot \Bigg( \sum_{j_2,j_2'}\sqrt{p_{j_2}p_{j_2'}} \big|Q_{j_2,j_2'}\big|\Bigg)\\
& = \sum_{j_2,j_2'}\sqrt{p_{j_2}p_{j_2'}} \big|Q_{j_2,j_2'}\big|.\\
\end{split}
\end{equation}
Lastly using the Cauchy-Schwartz inequality and the fact that the Hilbert-Schmidt norm of a $d\cross d$ unitary matrix is equal to $\sqrt{d}$, we can see that 
\begin{equation}
\begin{split}
    \sum_{j_2,j_2'}\sqrt{p_{j_2}p_{j_2'}} \big|Q_{j_2,j_2'}\big| &\leq \Bigg( \sum_{j_2,j_2'}p_{j_2}p_{j_2'}\Bigg)^{\frac{1}{2}}\Bigg(\sum_{j_2,j_2'}\big|Q_{j_2,j_2'}\big|^2\Bigg)^{\frac{1}{2}} \\
    &=\|Q\|_{HS}\\
    &= \sqrt{d}.\\
\end{split}
\end{equation}
The computation when substituting the term $\overline{Q_{j_1,j_2'}}\delta_{j_2,j_1'}$ is almost exactly the same so we omit it for brevity. We have thus showed that when $\rho_1(L)dL$ is at least a 2 design and $\rho_2(R)dR$ is an arbitrary distribution, the integral of Equation \eqref{eqn : sumc} is bounded in absolute value as   

\begin{footnotesize}
\begin{equation}\label{eqn : int2}
\begin{split}
    &\Bigg|\int \int 2\sum_{j_1,j_2,j_1',j_2'}\sqrt{p_{j_1}p_{j_1'}p_{j_2}p_{j_2'}}c_{j_1,j_1',j_2,j_2'}\bra{i}LR\ket{j_1}\overline{\bra{i}LR\ket{j_1'}}\bra{j_2'} R^*K_i R \ket{j_2} \rho_1(L)dL\rho_2(R)dR\Bigg|\\
    &\leq \int 2C\sqrt{d}\rho_2(R)dR \\
    &= \frac{2\sqrt{d}}{d(d+1)}.
\end{split}
\end{equation}
\end{footnotesize}
And since both $\avg\Bigg[2q_i\frac{\partial q_i}{\partial\theta}\Bigg]$ and $\avg\Bigg[\dt\tTr\Big(\tTr_\cA\Phi_i^2\Big)\Bigg]$ are both bounded by the above integral, it follows that

\begin{equation}\label{eqn : detailmean1}
    \avg\Bigg[2q_i\frac{\partial q_i}{\partial\theta}\Bigg] \sim \mathscr{O}\Bigg(\frac{1}{d^{\frac{3}{2}}}\Bigg) \text{ and } \avg\Bigg[\dt\tTr\Big(\tTr_\cA\Phi_i^2\Big)\Bigg] \sim \mathscr{O}\Bigg(\frac{1}{d^{\frac{3}{2}}}\Bigg).\\
\end{equation}
Therefore 
\begin{equation}\label{eqn : detailmean2}
    \avg\Bigg[\dt T_f(\rho;\theta)\Bigg] = \sum_i\Bigg( \avg\Bigg[2q_i\frac{\partial q_i}{\partial\theta}\Bigg]-\avg\Bigg[\dt\tTr\Big(\tTr_\cA\Phi_i^2\Big)\Bigg]\Bigg) \sim \mathscr{O}\Bigg(\frac{1}{\sqrt{d}}\Bigg).\\
\end{equation}
This shows that the expected gradients converge exponentially to 0 in the number of qubits when $L$ is at least a 2-design and the distribution of $R$ is arbitrary. 

Suppose now that the distribution of $R$ is at least a 2-design and the the distribution of $L$ is arbitrary. We will make almost the same change of variables as before with $W=LR$ and $dW=dR$. However we now define $Q = LVL^*$ so that 
\begin{equation}
\begin{split}
    \overline{\bra{i}LVR\ket{j_2'}} &= \overline{\bra{i}QW\ket{j_2'}}\\
    &= \sum_\alpha \overline{Q_{i,\alpha}}\text{ }\overline{W_{\alpha, j_2'}}.\\
\end{split}
\end{equation}
Now using the same steps as in the first case in Equations \eqref{eqn : q1} and \eqref{eqn : q2}, we arrive at the integral 
\begin{equation}
\begin{split}
    &\sum_\alpha \int \overline{Q_{i,\alpha}}\int W_{i,j_1}W_{i,j_2}\overline{W_{i,j_1'}}\text{ }\overline{W_{\alpha, j_2'}}dW\rho_1(L)dL\\
    &=\sum_\alpha \int \overline{Q_{i,\alpha}} C\delta_{i,\alpha}\Big(\delta_{j_1,j_1'}\delta_{j_2,j_2'} +\delta_{j_1,j_2'}\delta_{j_2,j_1'}\Big) \rho_1(L)dL \\
    &=  C\int \overline{Q_{i,i}}\Big(\delta_{j_1,j_1'}\delta_{j_2,j_2'} +\delta_{j_1,j_2'}\delta_{j_2,j_1'}\Big) \rho_1(L)dL.\\
\end{split}
\end{equation}
And again we now substitute one of the expressions containing deltas into Equation \eqref{eqn : sumc} and bound the expression. 
\begin{equation}
\begin{split}
     \sum_{j_1,j_2,j_1',j_2'}\sqrt{p_{j_1}p_{j_1'}p_{j_2}p_{j_2'}}c_{j_1,j_1',j_2,j_2'}\overline{Q_{i,i}}\delta_{j_1,j_1'}\delta_{j_2,j_2'} &= \sum_{j_1,j_2}p_{j_1}p_{j_2}c_{j_1,j_1,j_2,j_2}\overline{Q_{i,i}}.\\
\end{split}
\end{equation}
Since both the $c_{j_1,j_1,j_2,j_2}$ and $\overline{Q_{i,i}}$ are bounded in absolute value by 1, we see that 
\begin{footnotesize}
\begin{equation}
\begin{split}
     &\Bigg|\int \int 2\sum_{j_1,j_2,j_1',j_2'}\sqrt{p_{j_1}p_{j_1'}p_{j_2}p_{j_2'}}c_{j_1,j_1',j_2,j_2'}\bra{i}LR\ket{j_1}\overline{\bra{i}LR\ket{j_1'}}\bra{j_2'} R^*K_i R \ket{j_2} \rho_1(L)dL\rho_2(R)dR\Bigg|\\
    &= \Bigg|\sum_i 2C\sum_{j_1,j_2}p_{j_1}p_{j_2}c_{j_1,j_1,j_2,j_2}\overline{Q_{i,i}} \Bigg| \\
    &\leq 2C\sum_i\sum_{j_1,j_2}p_{j_1}p_{j_2}\big|c_{j_1,j_1,j_2,j_2}\big| \big|Q_{i,i}\big|\\
    &\leq 2C\sum_i\sum_{j_1,j_2}p_{j_1}p_{j_2}\\
    &= 2C\sum_i 1\\
    &= \frac{2d}{d(d+1)} \\
    &= \frac{2}{d+1}.\\
\end{split}
\end{equation}
\end{footnotesize}
Now we the same steps as in Equations \eqref{eqn : detailmean1} and \eqref{eqn : detailmean2} to compute the expected value again. Thus in the case when $\rho_2(R)dR$ is at least a 2-design and $\rho_1(L)dL$ is given by an arbitrary distribution, we see that 
\begin{equation}
    \avg\Bigg[\dt T_f(\rho;\theta)\Bigg]  \sim \mathscr{O}\Bigg(\frac{1}{d}\Bigg).
\end{equation}

We now need to compute the second moment of the $\dt T_f(\rho;\theta)$. We will accomplish this by squaring the expression in Equation \eqref{eqn : deriv} and integrating term by term just as we did with mean. In order to compute these integrals, we need stronger,but still modest, assumptions on the distributions of $\rho_1$ amd $\rho_2$. Namely we must split into cases when $\rho_1$ is a 4-design and $\rho_2$ is not, and vice versa. Note that such a circuit is indeed possible in practice and only needs to be of polynomial depth \cite{kdes}. 

First let's square the sum in Equation \eqref{eqn : deriv} to understand what we're working with. 
\begin{footnotesize}

\begin{equation}\label{eqn : 2moment}
\begin{split}
    \Big(\dt T_f(\rho;\theta)\Big)^2&=4\Bigg(\sum_i \Bigg( q_i \frac{\partial q_i}{\partial\theta}-\tTr\Big(\tTr_\cA\Phi_i \frac{\partial \tTr_\cA \Phi_i}{\partial \theta}\Big)\Bigg)\Bigg)^2\\
    &= 4\sum_{i,j}\Bigg( q_i \frac{\partial q_i}{\partial\theta}-\tTr\Big(\tTr_\cA\Phi_i \frac{\partial \tTr_\cA \Phi_i}{\partial \theta}\Big)\Bigg)\Bigg( q_j \frac{\partial q_j}{\partial\theta}-\tTr\Big(\tTr_\cA\Phi_j \frac{\partial \tTr_\cA \Phi_j}{\partial \theta}\Big)\Bigg).
\end{split}
\end{equation}
\end{footnotesize}
When multiplying term by term in the last line of the equation above, making substitutions using Equations \eqref{eqn : qi2} through \eqref{eqn : dphii} and using the same change of variables $W=LR$, we encounter expressions of the form
\begin{footnotesize}

{\begin{equation} \label{eqn : sumc4}
\sum_{\substack{j_1,j_2,j_1',j_2' \\ j_3,j_4,j_3',j_4'}}\sqrt{p_{j_1}p_{j_1'}p_{j_2}p_{j_2'}p_{j_3}p_{j_3'}p_{j_4}p_{j_4'}}c_{j_1,j_1',j_2,j_2'}c_{j_3,j_3',j_4,j_4'}W_{i,j_1}\overline{W_{i,j_1'}}W_{i,j_3}\overline{W_{i,j_3'}} \bra{j_2'} R^*K_i R \ket{j_2} \bra{j_4'} R^*K_j R \ket{j_4}
\end{equation}}
\end{footnotesize}
for each each fixed $i$ and $j$. Just as we did in the computation of the mean, we will fix $i$ and $j$, compute the appropriate integrals for each term, and sum the results at the end. 

Again, we only need to worry about the terms of the above equation that involve only entries of $W$, $L$,  or $R$ when integrating, and can substitute the values of these integrals back into the Equation. Because there are now two commutators, $K_i$ and $K_j$, in the expression, we'll need to break our proof into cases based on which terms of the commutator are multiplied together. To understand the cases better, let's expand  fully expand $\bra{j_2'} R^*K_i R \ket{j_2} \bra{j_4'} R^*K_j R \ket{j_4}$. Using the definitions of $K_i$ and $k_j$ respectively, we can see that 

\begin{equation}\label{eqn : comex}
\begin{split}
    &\bra{j_2'} R^*K_i R \ket{j_2} \bra{j_4'} R^*K_j R \ket{j_4}\\
    &= -\bra{j_2'}R^* \Big[V,L^*\ketbra{i}L\Big]R \ket{j_2} \bra{j_4'}R^*\Big[V,L^*\ketbra{i}L\Big]R\ket{j_4}\\
    &= -\bra{j_2'}R^*VL^*\ketbra{i}LR\ket{j_2}\bra{j_4'}R^*VL^*\ketbra{i}LR\ket{j_4}\\
    &+ \bra{j_2'}R^*VL^*\ketbra{i}LR\ket{j_2}\bra{j_4'}R^*L^*\ketbra{i}LVR\ket{j_4}\\
    &+\bra{j_2'}R^*L^*\ketbra{i}LVR\ket{j_2}\bra{j_4'}R^*VL^*\ketbra{i}LR\ket{j_4}\\
    &- \bra{j_2'}R^*L^*\ketbra{i}LVR\ket{j_2}\bra{j_4'}R^*L^*\ketbra{i}LVR\ket{j_4}.
\end{split}
\end{equation}

This first case we need to consider, \emph{case A}, is when either the first term,
\begin{equation}
    -\bra{j_2'}R^*VL^*\ketbra{i}LR\ket{j_2}\bra{j_4'}R^*VL^*\ketbra{i}LR\ket{j_4},
\end{equation}
 or the last term
 \begin{equation}
     -\bra{j_2'}R^*L^*\ketbra{i}LVR\ket{j_2}\bra{j_4'}R^*L^*\ketbra{i}LVR\ket{j_4},
 \end{equation}
 are used in the computation of the integral of Equation \eqref{eqn : sumc4}. The computations involved in bounding the integral of those expression are almost idential to each other. The second case, \emph{case B}, is when either of the middle terms in the above expansion of the commutator, are used in the computation of the integral of Equation \eqref{eqn : sumc4}. Let's look at \emph{case A} first. The expression that we want to integrate is given by 
\begin{equation}
    W_{i,j_2}W_{j,j_4}\bra{j_2'}R^*VL^*\ket{i}\bra{j_4'}R^*VL^*\ket{j}.
\end{equation}
Thus we need to integrate expressions of the form: 
\begin{equation}
    W_{i,j_1}\overline{W_{i,j_1'}}W_{i,j_3}\overline{W_{i,j_3'}}W_{i,j_2}W_{j,j_4}\bra{j_2'}R^*VL^*\ket{i}\bra{j_4'}R^*VL^*\ket{j}.
\end{equation}
Since the computations that follow are very similar for when either $\rho_1(L)dL$ or $\rho_2(R)dR$ is a 4-design respectively, we only provide the details for the first case. So suppose that $\rho_1(L)dL$ is at least a 4-design , and again define $Q$ to be $R^*VR$ so that the integral of the above expression is equal to 
\begin{equation}
\begin{split}
    &\int\int W_{i,j_1}\overline{W_{i,j_1'}}W_{i,j_3}\overline{W_{i,j_3'}}W_{i,j_2}W_{j,j_4}\bra{j_2'}R^*VL^*\ket{i}\bra{j_4'}R^*VL^*\ket{j} \rho_1(L)dL\rho_2(R)dR\\
    &= \int\int W_{i,j_1}\overline{W_{i,j_1'}}W_{i,j_3}\overline{W_{i,j_3'}}W_{i,j_2}W_{j,j_4}\bra{j_2'}QW^*\ket{i}\bra{j_4'}QW^*\ket{j} \rho_1(L)dL\rho_2(R)dR\\
    &=\int\int W_{i,j_1}\overline{W_{i,j_1'}}W_{i,j_3}\overline{W_{i,j_3'}}W_{i,j_2}W_{j,j_4}\overline{\bra{i}WQ^*\ket{j_2'}}\overline{\bra{j}WQ^*\ket{j_4'}} \rho_1(L)dL\rho_2(R)dR\\
    &=\int\int W_{i,j_1}\overline{W_{i,j_1'}}W_{i,j_3}\overline{W_{i,j_3'}}W_{i,j_2}W_{j,j_4}\overline{\bra{i}WQ\ket{j_2'}}\overline{\bra{j}WQ\ket{j_4'}} \rho_1(L)dL\rho_2(R)dR\\
    &=\sum_{\alpha,\beta}\overline{Q_{\alpha, j_2'}}\text{ }\overline{Q_{\beta, j_4'}}\int\int W_{i,j_1}W_{i,j_2}W_{j,j_3}W_{j,j_4}\overline{W_{i,j_1'}}\text{ }\overline{W_{j,j_3'}}\text{ }\overline{W_{i,\alpha}}\text{ }\overline{W_{j,\beta}}\rho_1(L)dL \rho_2(R)dR 
\end{split}
\end{equation}
where the two lines are derived using the fact that $Q$ is hermitian, and by expanding $Q$ in terms of its coordinates. Since computing the exact value of this integral involves summing over $S_4\cross S_4$, we'll instead show the value of one of terms in the sum using $\sigma = (1234)$ and $\tau = (13)(24)$. It'll be easy to see that each term has value on the order $\frac{1}{d^4}$ and that the choice of permutation doesn't play a significant role in the computation. 

These choices of permutations lead us to the following expression:
\begin{equation}
    \wg\big(\sigma\tau^{-1}\big)\sum_{\alpha \beta}\overline{Q_{\alpha, j_2'}}\text{ }\overline{Q_{\beta, j_4'}}\delta_{i,j}\delta_{j_1, \alpha}\delta_{j_2, \beta}\delta_{j_3, j_1'}\delta_{j_4, j_3'} = \wg\big(\sigma\tau^{-1}\big)\overline{Q_{j_1, j_2'}}\text{ } \overline{Q_{j_2, j_4'}}\delta_{i,j}\delta_{j_3, j_1'}\delta_{j_4, j_3'}.
\end{equation}

Ignoring the constant $\wg\big(\sigma\tau^{-1}\big)$ and substituting the right hand side of the equation into Equation \eqref{eqn : sumc4} we obtain 
\begin{equation}
\begin{split}
    &\sum_{\substack{j_1,j_2,j_1',j_2' \\ j_3,j_4,j_3',j_4'}}\sqrt{p_{j_1}p_{j_1'}p_{j_2}p_{j_2'}p_{j_3}p_{j_3'}p_{j_4}p_{j_4'}}c_{j_1,j_1',j_2,j_2'}c_{j_3,j_3',j_4,j_4'}\overline{Q_{j_1, j_2'}}\text{ }\overline{Q_{j_2, j_4'}}\delta_{j_3, j_1'}\delta_{j_4, j_3'}\delta_{i,j} \\
    &=\sum_{\substack{j_1,j_2,j_2' \\ j_3,j_4,j_4'}}p_{j_3}p_{j_4}\sqrt{p_{j_1}p_{j_2}p_{j_2'}p_{j_4'}}c_{j_1,j_3,j_2,j_2'}c_{j_3,j_4,j_4,j_4'}\overline{Q_{j_1, j_2'}}\text{ }\overline{Q_{j_2, j_4'}}\delta_{i,j}.
\end{split}
\end{equation}
Then ignoring the $\delta_{i,j}$ term for the moment, taking absolute values, and using the Cuachy-Schwartz inequality again, the above sum is less than or equal to 
\begin{equation}
\begin{split}
    &\sum_{\substack{j_1,j_2,j_2' \\ j_3,j_4,j_4'}}p_{j_3}p_{j_4}\sqrt{p_{j_1}p_{j_2}p_{j_2'}p_{j_4'}}|Q_{j_1, j_2'}||Q_{j_2, j_4'}| \\
    & = \sum_{j_1,j_2,j_2',j_4'}\sqrt{p_{j_1}p_{j_2}p_{j_2'}p_{j_4'}}|Q_{j_1, j_2'}||Q_{j_2, j_4'}| \\
    &=\Big(\sum_{j_1,j_2'}\sqrt{p_{j_1}p_{j_2'}}|Q_{j_1, j_2'}|\Big)\Big( \sum_{j_2,j_4'}\sqrt{p_{j_2}p_{j_4'}}|Q_{j_2, j_4'}|\Big)\\
    &\leq \Big(\sum_{j_1,j_2'}p_{j_1}p_{j_2'}\Big)^\frac{1}{2} \Big(\sum_{j_1,j_2'}|Q_{j_1, j_2'}|^2\Big)^\frac{1}{2}\Big(\sum_{j_2,j_4'} p_{j_2}p_{j_4'}\Big)^\frac{1}{2}\Big(\sum_{j_2,j_4'}|Q_{j_2, j_4'}|^2\Big)^\frac{1}{2}\\
    &= \|Q\|_{HS}^2\\
    &=d.
\end{split}
\end{equation}
Choosing any other pair of permutation $\sigma$ and $\tau$ will yield a result of the same magnitude. Since $\wg\big(\sigma\tau^{-1}\big)$ is of order at most $\frac{1}{d^4}$ for any $\sigma,\tau\in \text{S}_4$, the integral of these terms must be of order at most $\frac{1}{d^3}$. The computation for when $R$ is a 4-design will be similar to the previous case and will result in a value whose magnitude is of the order $\frac{1}{d^4}$.

Now we need to consider \emph{case B}, in which we substitute the middle terms from the commutator in Equation \eqref{eqn : comex} into Equation \eqref{eqn : sumc4} and integrate. Namely, we want to integrate expressions of the form 
\begin{equation}
     W_{i,j_1}\overline{W_{i,j_1'}}W_{i,j_3}\overline{W_{i,j_3'}}\bra{j_2'}R^*VL^*\ketbra{i}LR\ket{j_2}\bra{j_4'}R^*L^*\ketbra{i}LVR\ket{j_4}.\\
\end{equation}

Now suppose that $L$ is at least a 4-design and define $Q=R^*VR$ again. Then the above expression reduces to 
\begin{equation}
\begin{split}
     &W_{i,j_1}W_{i,j_2}W_{i,j_3}\overline{W_{i,j_1'}}\text{ }\overline{W_{i,j_3'}}\text{ }\overline{W_{j,j_4'}}\bra{j_2'}QW^*\ket{i}\bra{j}WQ\ket{j_4}\\
     &=\sum_{\alpha,\beta}\overline{Q_{j_2',\alpha}}Q_{\beta,j_4}W_{i,j_1}W_{i,j_2}W_{i,j_3}W_{j,\beta}\overline{W_{i,j_1'}}\text{ }\overline{W_{i,j_3'}}\text{ }\overline{W_{j,j_4'}}\text{ }\overline{W_{i,\alpha}}.
\end{split} 
\end{equation}
This integral is a little trickier than the rest. Up to now the indices $\alpha$ and $\beta$, that were added when applying $Q$ to a ket or bra, have been paired with a $j$,$j'$, or an $i$ when integrating. However, there is now the possibility that $\alpha$ and $\beta$ can be paired with each other this time. When looking at the string of Kronecker deltas produced from the Haar integral, we must simplify the sum in Equation \eqref{eqn : sumc4} based on two distinct cases. The first, when $\alpha$ and $\beta$ are not paired with each other, and the second when $\alpha$ and $\beta$ are paired with each other. Computationally, the first case is nearly identical to \emph{case A}. We will only focus on the case in which $\alpha$ and $\beta$ are paired with each other in the integration formula in Lemma~\ref{eqn : wein}. Once again, we only focus on example permutations of this form, and leave it to the reader to see that the computations are unchanged by a different permutation as long as $\alpha$ and $\beta$ are still paired with each other. For simplicity, take $\sigma=\tau=(1)_4$. Then integrating the above expression with respect to $W$ yields 
\begin{equation}
\begin{split}
     &\int\sum_{\alpha,\beta}\overline{Q_{j_2,'\alpha}}Q_{\beta,j_4}W_{i,j_1}W_{i,j_2}W_{i,j_3}W_{j,\beta}\overline{W_{i,j_1'}}\text{ }\overline{W_{i,j_3'}}\text{ }\overline{W_{j,j_4'}}\text{ }\overline{W_{i,\alpha}}dW\\
     &=\wg(1,1,1,1)\sum_{\alpha,\beta}\overline{Q_{j_2',\alpha}}Q_{\beta,j_4}\delta_{j_1, j_1'}\delta_{j_2, j_3'}\delta_{j_3, j_4'}\delta_{\alpha, \beta}\\
     &=\wg(1,1,1,1)\sum_{\alpha}\overline{Q_{j_2',\alpha}}Q_{\alpha,j_4}\delta_{j_1 ,j_1'}\delta_{j_2, j_3'}\delta_{j_3, j_4'}.\\
\end{split}
\end{equation}
Just like before, we ignore the leading constant for readability and substitute the above expression into \eqref{eqn : sumc4} and bound it in absolute value. 

\begin{equation}
\begin{split}
    &\sum_\alpha\sum_{\substack{j_1,j_2,j_1',j_2' \\ j_3,j_4,j_3',j_4'}}\sqrt{p_{j_1}p_{j_1'}p_{j_2}p_{j_2'}p_{j_3}p_{j_3'}p_{j_4}p_{j_4'}}c_{j_1,j_1',j_2,j_2'}c_{j_3,j_3',j_4,j_4'}\overline{Q_{j_2',\alpha}}Q_{\alpha,j_4}\delta_{j_1, j_1'}\delta_{j_2, j_3'}\delta_{j_3, j_4'}\\
    &= \sum_\alpha \sum_{j_1,j_2,j_3,j_4,j_2'}p_{j_1}p_{j_2}p_{j_3}\sqrt{p_{j_2'}p_{j_4}}c_{j_1,j_1,j_2,j_2'}c_{j_3,j_2,j_4,j_3}\overline{Q_{j_2',\alpha}}Q_{\alpha,j_4}.
\end{split}
\end{equation}

Taking absolute values using Cauchy Schwartz again, we see that the magnitude of this sum is bounded by 
\begin{equation}
\begin{split}
    &\sum_\alpha\sum_{j_4 j_2'}\sqrt{p_{j_2'}p_{j_4}}|Q_{j_2',\alpha}||Q_{\alpha,j_4}|\\
    &= \sum_\alpha \Big( \sum_{j_4}\sqrt{p_{j_4}}|Q_{\alpha,j_4}|\Big)\Big(\sum_{j_2'}\sqrt{p_{j_2'}}|Q_{j_2',\alpha}|\Big)\\
    &\leq \sum_\alpha \Big(\sum_{j_4}p_{j_4}\Big)^\frac{1}{2}\Big(\sum_{j_4}|Q_{\alpha,j_4}|^2\Big)^\frac{1}{2}\Big(\sum_{j_2'}p_{j_2'}\Big)^\frac{1}{2} \Big(\sum_{j_2'}|Q_{\alpha,j_2'}|^2\Big)^\frac{1}{2}\\
    &\leq \sum_\alpha 1\\
    &= d.\\
\end{split}
\end{equation}
Thus when multiplying by the appropriate evaluation of the Weingarten function, this expression is of order at most $\frac{1}{d^3}$. And since there are $d^2$ terms in the integrand of the second moment given in Equation \eqref{eqn : 2moment}, the second moment must be of order $\frac{1}{d}$. And therefore 
\begin{equation}
    \variance\Big[\dt T_f(\rho;\theta)\Big] \sim \mathscr{O}\Big(\frac{1}{d}\Big).
\end{equation}
Lastly, since $d=2^k$ where $k$ is the number of qubits, we see that the variance decreases exponentially in the number of qubits, implying the existence of barren plateaus in the training landscape of $T_f(\rho;\theta)$ when using the ansatz in Equation \eqref{eqn : ansatz} with sufficient depth.


\end{document}